\documentclass[12pt]{article}
\usepackage{amssymb}
\usepackage{multirow}
\usepackage{epsfig,rotating}
\usepackage{amssymb,amsmath}
\usepackage{latexsym}
\usepackage{tabularx}
\usepackage{array}
\usepackage{natbib}
\usepackage[retainorgcmds]{IEEEtrantools}
\usepackage{url}            
\usepackage{booktabs}       
\usepackage{amsfonts}       
\usepackage{nicefrac}       
\usepackage{microtype}      
\usepackage{lipsum}
\usepackage{amsmath}
\usepackage{amssymb}
\usepackage{indentfirst}
\usepackage{mathtools}
\usepackage{algorithm}
\usepackage{algpseudocode}
\usepackage{appendix}
\usepackage{hyperref} 
\hypersetup{
	colorlinks   = true, 
	urlcolor     = blue, 
	linkcolor    = blue, 
	citecolor   = blue 
}
\usepackage{nccmath}
\usepackage{centernot}
\usepackage{setspace}

\renewcommand{\baselinestretch}{1.5}

\usepackage{tkz-graph}  
\usetikzlibrary{shapes.geometric}
\usepackage{siunitx}
\usetikzlibrary{positioning}
\tikzset{mynode/.style={draw,text width=1in,align=center}
}

\usepackage{bm,fullpage,enumerate,setspace}
	
\newcommand{\indep}{\perp \!\!\! \perp}
\newcommand{\nindep}{\not\!\perp\!\!\!\perp}

\begin{document}

\def\reduce#1{{\small  #1}}
\def\c#1{\ensuremath{\mathcal{#1}}}
\def\foot#1{\footnote{#1}}

\newcommand{\nmathbf}{\bm}

\def\bfA{\nmathbf A}
\def\bfB{\nmathbf B}
\def\bfC{\nmathbf C}
\def\bfD{\nmathbf D}
\def\bfE{\nmathbf E}
\def\bfF{\nmathbf F}
\def\bfG{\nmathbf G}
\def\bfH{\nmathbf H}
\def\bfI{\nmathbf I}
\def\bfJ{\nmathbf J}
\def\bfK{\nmathbf K}
\def\bfL{\nmathbf L}
\def\bfM{\nmathbf M}
\def\bfN{\nmathbf N}
\def\bfO{\nmathbf O}
\def\bfP{\nmathbf P}
\def\bfQ{\nmathbf Q}
\def\bfR{\nmathbf R}
\def\bfS{\nmathbf S}
\def\bfT{\nmathbf T}
\def\bfU{\nmathbf U}
\def\bfV{\nmathbf V}
\def\bfW{\nmathbf W}
\def\bfX{\nmathbf X}
\def\bfY{\nmathbf Y}
\def\bfZ{\nmathbf Z}

\def\bfa{\nmathbf a}
\def\bfb{\nmathbf b}
\def\bfc{\nmathbf c}
\def\bfd{\nmathbf d}
\def\bfe{\nmathbf e}
\def\bff{\nmathbf f}
\def\bfg{\nmathbf g}
\def\bfh{\nmathbf h}
\def\bfi{\nmathbf i}
\def\bfj{\nmathbf j}
\def\bfk{\nmathbf k}
\def\bfl{\nmathbf l}
\def\bfm{\nmathbf m}
\def\bfn{\nmathbf n}
\def\bfo{\nmathbf o}
\def\bfp{\nmathbf p}
\def\bfq{\nmathbf q}
\def\bfr{\nmathbf r}
\def\bfs{\nmathbf s}
\def\bft{\nmathbf t}
\def\bfu{\nmathbf u}
\def\bfv{\nmathbf v}
\def\bfw{\nmathbf w}
\def\bfx{\nmathbf x}
\def\bfy{\nmathbf y}
\def\bfz{\nmathbf z}

\def\bfalpha  {\nmathbf \alpha}
\def\bfbeta   {\nmathbf \beta}
\def\bfgamma  {\nmathbf \gamma}
\def\bfdelta  {\nmathbf \delta}
\def\bfepsilon{\nmathbf \epsilon}
\def\bfzeta   {\nmathbf \zeta}
\def\bfeta    {\nmathbf \eta}
\def\bftheta  {\nmathbf \theta}
\def\bfiota   {\nmathbf \iota}
\def\bfkappa  {\nmathbf \kappa}
\def\bflambda {\nmathbf \lambda}
\def\bfmu     {\nmathbf \mu}
\def\bfnu     {\nmathbf \nu}
\def\bfxi     {\nmathbf \xi}
\def\bfomicron{\nmathbf \omicron}
\def\bfpi     {\nmathbf \pi}
\def\bfrho    {\nmathbf \rho}
\def\bfsigma  {\nmathbf \sigma}
\def\bftau    {\nmathbf \tau}
\def\bfupsilon{\nmathbf \upsilon}
\def\bfphi    {\nmathbf \phi}
\def\bfpsi    {\nmathbf \psi}
\def\bfchi    {\nmathbf \chi}
\def\bfomega  {\nmathbf \omega}

\def\bfAlpha  {\nmathbf \Alpha}
\def\bfBeta   {\nmathbf \Beta}
\def\bfGamma  {\nmathbf \Gamma}
\def\bfDelta  {\nmathbf \Delta}
\def\bfEpsilon{\nmathbf \Epsilon}
\def\bfZeta   {\nmathbf \Zeta}
\def\bfEta    {\nmathbf \Eta}
\def\bfTheta  {\nmathbf \Theta}
\def\bfIota   {\nmathbf \Iota}
\def\bfKappa  {\nmathbf \Kappa}
\def\bfLambda {\nmathbf \Lambda}
\def\bfMu     {\nmathbf \Mu}
\def\bfNu     {\nmathbf \Nu}
\def\bfXi     {\nmathbf \Xi}
\def\bfOmicron{\nmathbf \Omicron}
\def\bfPi     {\nmathbf \Pi}
\def\bfRho    {\nmathbf \Rho}
\def\bfSigma  {\nmathbf \Sigma}
\def\bfTau    {\nmathbf \Tau}
\def\bfUpsilon{\nmathbf \Upsilon}
\def\bfPhi    {\nmathbf \Phi}
\def\bfPsi    {\nmathbf \Psi}
\def\bfChi    {\nmathbf \Chi}
\def\bfOmega  {\nmathbf \Omega}

\newcommand{\ttheta}{\tilde{\theta}}
\newcommand{\bfzero}{{\nmathbf 0}}
\newcommand{\bfone}{{\nmathbf 1}}
\newcommand{\vareps}{\varepsilon}
\def\bfvareps{\nmathbf \varepsilon}
\newcommand{\tgamma}{\tilde\gamma}

\newcommand{\cfA}{\mbox{\c{A}}}
\newcommand{\cfB}{\mbox{\c{B}}}
\newcommand{\cfC}{\mbox{\c{C}}}
\newcommand{\cfD}{\mbox{\c{D}}}
\newcommand{\cfE}{\mbox{\c{E}}}
\newcommand{\cfF}{\mbox{\c{F}}}
\newcommand{\cfG}{\mbox{\c{G}}}
\newcommand{\cfH}{\mbox{\c{H}}}
\newcommand{\cfI}{\mbox{\c{I}}}
\newcommand{\cfJ}{\mbox{\c{J}}}
\newcommand{\cfK}{\mbox{\c{K}}}
\newcommand{\cfL}{\mbox{\c{L}}}
\newcommand{\cfM}{\mbox{\c{M}}}
\newcommand{\cfN}{\mbox{\c{N}}}
\newcommand{\cfO}{\mbox{\c{O}}}
\newcommand{\cfP}{\mbox{\c{P}}}
\newcommand{\cfQ}{\mbox{\c{Q}}}
\newcommand{\cfR}{\mbox{\c{R}}}
\newcommand{\cfS}{\mbox{\c{S}}}
\newcommand{\cfT}{\mbox{\c{T}}}
\newcommand{\cfU}{\mbox{\c{U}}}
\newcommand{\cfV}{\mbox{\c{V}}}
\newcommand{\cfX}{\mbox{\c{X}}}
\newcommand{\cfY}{\mbox{\c{Y}}}
\newcommand{\cfZ}{\mbox{\c{Z}}}

\def\boldfacefake#1{\kern-4pt
   \hbox{ \mathsurround=0pt
   \hbox to 0.4pt{$#1$\hss}\hbox to 0.4pt{$#1$\hss}\hbox {$#1$}}}

\def\bfitI{\mbox{\boldfacefake{\it I}}}

\newcommand{\bcfA}{\boldsymbol{\mathcal{A}}}
\newcommand{\bcfB}{\boldsymbol{\mathcal{B}}}
\newcommand{\bcfC}{\boldsymbol{\mathcal{C}}}
\newcommand{\bcfD}{\boldsymbol{\mathcal{D}}}
\newcommand{\bcfE}{\boldsymbol{\mathcal{E}}}
\newcommand{\bcfF}{\boldsymbol{\mathcal{F}}}
\newcommand{\bcfG}{\boldsymbol{\mathcal{G}}}
\newcommand{\bcfH}{\boldsymbol{\mathcal{H}}}
\newcommand{\bcfI}{\boldsymbol{\mathcal{I}}}
\newcommand{\bcfJ}{\boldsymbol{\mathcal{J}}}
\newcommand{\bcfK}{\boldsymbol{\mathcal{K}}}
\newcommand{\bcfL}{\boldsymbol{\mathcal{L}}}
\newcommand{\bcfM}{\boldsymbol{\mathcal{M}}}
\newcommand{\bcfN}{\boldsymbol{\mathcal{N}}}
\newcommand{\bcfO}{\boldsymbol{\mathcal{O}}}
\newcommand{\bcfP}{\boldsymbol{\mathcal{P}}}
\newcommand{\bcfQ}{\boldsymbol{\mathcal{Q}}}
\newcommand{\bcfR}{\boldsymbol{\mathcal{R}}}
\newcommand{\bcfS}{\boldsymbol{\mathcal{S}}}
\newcommand{\bcfT}{\boldsymbol{\mathcal{T}}}
\newcommand{\bcfU}{\boldsymbol{\mathcal{U}}}
\newcommand{\bcfV}{\boldsymbol{\mathcal{V}}}
\newcommand{\bcfW}{\boldsymbol{\mathcal{W}}}
\newcommand{\bcfX}{\boldsymbol{\mathcal{X}}}
\newcommand{\bcfY}{\boldsymbol{\mathcal{Y}}}
\newcommand{\bcfZ}{\boldsymbol{\mathcal{Z}}}


\newcommand{\p}{\mbox{P}}
\newcommand{\D}{\mbox{D}}
\newcommand{\E}{\mbox{E}}
\newcommand{\Mo}{\mbox{Mo}}
\newcommand{\Me}{\mbox{Me}}
\newcommand{\Cov}{\mbox{Cov}}
\newcommand{\Var}{\mbox{Var}}
\newcommand{\Corr}{\mbox{Corr}}
\newcommand{\Q}{\mbox{Q}}
\newcommand{\vech}{\mbox{vech}}
\newcommand{\tr}{\mbox{tr}}
\newcommand{\supp}{\mbox{supp}}
\newcommand{\dto}{\stackrel{d}{\to}}

\newcommand{\Bb}{\mbox{Bb}}
\newcommand{\Be}{\mbox{Be}}
\newcommand{\Bi}{\mbox{Bi}}
\newcommand{\Br}{\mbox{Br}}
\newcommand{\Ca}{\mbox{Ca}}
\newcommand{\Di}{\mbox{Di}}
\newcommand{\Ex}{\mbox{Ex}}
\newcommand{\Fs}{\mbox{Fs}}
\newcommand{\Ga}{\mbox{Ga}}
\newcommand{\Ge}{\mbox{Ge}}
\newcommand{\Gg}{\mbox{Gg}}
\newcommand{\Hy}{\mbox{Hy}}
\newcommand{\Ig}{\mbox{Ig}}
\newcommand{\Ip}{\mbox{Ip}}
\newcommand{\Lo}{\mbox{Lo}}
\newcommand{\Mu}{\mbox{Mu}}
\newcommand{\Nb}{\mbox{Nb}}
\newcommand{\Ng}{\mbox{Ng}}
\newcommand{\Nw}{\mbox{Nw}}
\newcommand{\Po}{\mbox{Po}}
\newcommand{\Pg}{\mbox{Pg}}
\newcommand{\Pn}{\mbox{Pn}}
\newcommand{\Ra}{\mbox{Ra}}
\newcommand{\St}{\mbox{St}}
\newcommand{\Un}{\mbox{Un}}
\newcommand{\Wi}{\mbox{Wi}}

\newcommand{\dd}[1]{\,d#1}
\newcommand{\barx}{\mbox{$\overline x$}}
\newcommand{\comb}[2]{{#1\choose#2}}
\newcommand{\ontop}[2]{{#1\atop#2}}
\newcommand{\h}{\hbox{$1\over2$}}
\newcommand{\ok}{\hfill\fbox{}}

\newcommand{\brow}[2]{\mbox{$\{{#1}_1,\ldots,{#1}_{#2}\}$}}
\newcommand{\prow}[2]{\mbox{$({#1}_1,\ldots,{#1}_{#2})$}}
\newcommand{\row}[2]{\mbox{${#1}_1,\ldots,{#1}_{#2}$}}
\newcommand{\data}{\row{x}{n}}
\newcommand{\bdata}{\prow{x}{n}}
\newcommand{\ie}{\emph{i.e.},\ }
\newcommand{\co}{\emph{cf.}\ }
\newcommand{\eg}{\emph{e.g.}, }
\newcommand{\etalc}{\emph{et al.},\ }
\newcommand{\etal}{\emph{et al.}\ }

\newenvironment{mat}{\left(\begin{array}}{\end{array}\right)}

\newcommand{\twomat}[5]{\begin{array}{ll}\displaystyle
              #1&\mbox{#2}\\[#5pt]#3&\mbox{#4}\end{array}}

\newcommand{\twocases}[6]
             {#1=\left\{\twomat{#2}{#3}{#4}{#5}{#6}\right.}

\newcommand{\mymatrix}[4]
           {\left(\begin{array}{ll}{#1} & {#2}\\
                                   {#3} & {#4}
                    \end{array} \right)}

\newcommand{\btable}{\begin{table}[h]\centering}
\newcommand{\etable}{\end{table}}
\newcommand{\bt}{\begin{parag}\small \let\b=\nsb \let\sb=\nssb \begin{tabular}}
\newcommand{\et}{\end{tabular}\let\b=\nb \let\sb=\nsb\end{parag}}
\newcommand{\capt}[1]
         {\begin{quotation}\caption{{\small #1 }}\vs{-2}\end{quotation}}

\newenvironment{parag}{\par}{\par}
\newenvironment{dif}
    {\begin{parag}\small \let\b=\nsb \let\sb=\nssb \begin{parag}}
    {\let\b=\nb \let\sb=\nsb \end{parag}\end{parag}}


\newcommand{\be}{\begin{eqnarray}}
\newcommand{\ee}{\end{eqnarray}}
\newcommand{\ba}{\begin{eqnarray*}}
\newcommand{\ea}{\end{eqnarray*}}

\newcommand{\go}{\rightarrow}
\newcommand{\goi}{\rightarrow \infty}
\newcommand{\ol}{\overline}
\newcommand{\fr}{\frac}
\newcommand{\pn}{\par\noindent}
\newcommand{\nc}{\nonumber\\}
\newcommand{\ssum}{\mbox{$\sum$}}
\newcommand{\hhline}{\hline\hline}

\newtheorem{theorem0}{Theorem}
\newtheorem{lemma0}{Lemma}
\newtheorem{remark0}{Remark}
\newtheorem{fact0}{Fact}
\newtheorem{example0}{Example}
\newtheorem{definition0}{Definition}
\newtheorem{corollary0}{Corollary}
\newtheorem{proposition0}{Proposition}
\newtheorem{algorithmY}{Algorithm}

\newenvironment{theorem}{\begin{theorem0} \mbox{} }{\end{theorem0}}
\newenvironment{lemma}{\begin{lemma0} \mbox{}}{\end{lemma0}}
\newenvironment{remark}{\begin{remark0} \mbox{}}{\end{remark0}}
\newenvironment{fact}{\begin{fact0} \mbox{}}{\end{fact0}}
\newenvironment{example}{\begin{example0} }{\end{example0}}
\newenvironment{definition}{\begin{definition0} \mbox{}}{\end{definition0}}
\newenvironment{corollary}{\begin{corollary0} \mbox{} }{\end{corollary0}}
\newenvironment{proposition}{\begin{proposition0}\mbox{} }{\end{proposition0}}
\newenvironment{algorithm1}{\begin{algorithmY}\mbox{} }{\end{algorithmY}}

\newcommand{\reals}{\mbox{\rm I\kern-.20em R}}
\newcommand{\sreals}{\mbox{\small \rm I\kern-.20em R}}
\newcommand{\mylinel}{\renewcommand{\baselinestretch}{1.8}\tiny\small}
\newcommand{\goto}{\rightarrow}
\newcommand{\expect}{\E}

\newcommand{\bdfn}{\begin{dfn}}
\newcommand{\edfn}{\end{dfn}}
\newcommand{\bteo}{\begin{teo}}
\newcommand{\eteo}{\end{teo}}
\newcommand{\bexa}{\begin{exa}}
\newcommand{\eexa}{\end{exa}}
\newcommand{\bdif}{\begin{dif}}
\newcommand{\edif}{\end{dif}}
\newcommand{\bpro}{\begin{proof}}
\newcommand{\epro}{\end{proof}}


\begin{center}
{\bf\Large Design-Assisted Regression
\\\vspace*{0.1in} \vspace*{0.1in}}
\end{center}

\begin{center}
Shangyuan Ye$^{1,\dag}$, Guanbo Wang$^{2}$, Cong Zhang$^{3}$, Ye Liang$^{4}$
\vspace*{0.15in}

$^{1}$Department of Mathematics and Statistics, Florida International University \\
$^2$The Dartmouth Institute for Health Policy and Clinical Practice, Geisel School of Medicine, Dartmouth College \\
$^3$School of Economics, University of Nottingham Ningbo China \\
$^4$Department of Statistics, Oklahoma State University \\
\vspace*{0.3in}

\end{center}




\begin{quotation}
\small

\begin{spacing}{1}
\noindent {\bf Abstract}
We consider regression problems in which the marginal distribution of the covariates is informative for estimation and variable selection, rather than merely auxiliary. Motivated by random-design, high-dimensional, and latent-effect settings, we propose a general design-assisted regression framework in which the estimating criterion depends on both the conditional model for $Y \mid \bfX$ and structured features of the covariate distribution. The framework identifies two roles of design information: stabilizing weak design directions through quadratic regularization and correcting latent-effect distortion through nuisance augmentation. We establish oracle properties for the resulting estimator, separate the effects of stochastic error, shrinkage, and approximation, and compare it with a benchmark sparse procedure that ignores design information. These results show that the proposed framework improves estimation while preserving first-order prediction performance. Numerical studies and two real-data applications illustrate the practical impact of incorporating design information.

\vspace*{0.5in}

\vspace*{0.15in}

\end{spacing}

\end{quotation}

\vspace*{0.3in}

$^{\dag}$ Corresponding author. E-mail: \textit{sye@fiu.edu}

\newpage
\setcounter{page}{1}

\section{Introduction}
In classical regression, inference is usually based on the conditional law of $Y \mid \bfX$, and the marginal distribution of the covariates is treated as auxiliary \citep{fraser2004ancillaries}. This viewpoint is natural in low-dimensional parametric models, where the target parameter is defined through the conditional mean or likelihood alone. However, in modern data settings, the same scientific question may be studied using data sets that differ substantially in size, quality, and heterogeneity. In particular, fitting the same regression model to a small but relatively homogeneous data set and to a large, heterogeneous observational data set can lead to markedly different results, even when the underlying structural relationship is assumed to be the same.

One reason is that larger observational data sets often contain richer latent structures, such as cluster effects, hidden confounding, selection effects, and regime heterogeneity, for which the covariate distribution itself becomes informative for estimation rather than merely incidental. This observation is related to the distinction between explanatory and predictive modeling emphasized by \citet{shmueli2010explain}, who argued that predictive success and faithful recovery of the scientific target are conceptually different goals. In a similar spirit, \citet{bzdok2020inference} discussed how predictive and inferential objectives can diverge in modern biomedical data analysis. For regression models with latent heterogeneity, \citet{heagerty2000marginalized} showed that conditional and marginal regression parameters can differ substantially so that the same underlying structural model may induce different fitted regression targets depending on how latent effects are handled.

This perspective also helps to explain why modern machine-learning algorithms can perform well for prediction in large observational data sets, but their ability for structural interpretation may be unreliable. When the covariate distribution carries information about latent effects, a predictive algorithm often exploits this information to improve its risk criteria, even though part of that predictive signal is scientifically irrelevant to the parameter of interest. In other words, some covariates may be useful because they proxy for latent variation rather than because they represent the structural mechanism. This phenomenon is also closely related to the literature on distribution shift, where the design distribution itself affects predictive performance \citep{sugiyama2007covariate}. From this viewpoint, more data are not necessarily translated into more scientific knowledge, as larger and more heterogeneous data sets, which may improve prediction, are more likely to lead to target distortion if the design distribution is ignored.

Incorporating covariate distribution information into regression is closely related to shrinkage estimators and other methods that exploit the geometry of the design. Ridge and generalized ridge regression use this information to stabilize weakly identified directions and improve efficiency \citep{tikhonov1977illposed}, while principal component regression and envelope methods rely on low-dimensional structure in the predictor distribution for a similar purpose \citep{jolliffe1982pcr,cook2010envelope}. The synthetic heterogeneous-effects LASSO uses design features associated with random intercepts for high-dimensional clustered data analysis \citep{ye2026synthetic}.  For elliptically distributed $\bfX$, sufficient dimension reduction methods such as sliced inverse regression identify a basis for the central subspace that preserves the conditional distribution of $Y | \bfX$ \citep{li1991sliced}. Model-X knockoffs use the distribution of $\bfX$ to construct reference variables for valid selective inference rather than point estimation \citep{candes2018panning}. 

In this paper, we propose a general framework, namely the \textit{design-assisted regression}, in which estimation and inference depend not only on the conditional model for $Y \mid \bfX$, but also on structured information in the marginal distribution of $\bfX$. The framework highlights how design information can be used both to improve efficiency and to reduce target distortion. The framework also unifies multiple existing methods as special cases and motivates new methodology for modern regression problems. Theoretically, we establish oracle properties for the proposed estimator under this framework, quantify the contribution of latent-effect approximation and quadratic regularization to the estimation error, and also compare the proposed work with a benchmark sparse procedure that ignores design information. We provide a theoretical answer to the question of how design information can improve structural estimation while preserving first-order predictive performance. The results of this paper bring a new perspective for improving estimation in a broad class of regression problems, including clustered data analysis \citep{ye2025variable,ye2026synthetic}, causal inference with hidden confounding \citep{zivich2023proximal,miao2018proxy}, and modern settings involving heterogeneous effects or distribution shift \citep{kunzel2019metalearners,sugiyama2007covariate}.

\section{Design-assisted regression} \label{sec:method}
\subsection{Notations}
Throughout the article, the superscript $0$ is used to denote the true value of a given parameter. Let $\bfZ_i = (Y_i, \bfX_i)$ for $i = 1, \ldots, n$ be the observed data variables, where $Y_i \in \reals$ is the outcome variable and $\bfX_i = (X_{1,i}, \ldots, X_{p,i})^\top \in \reals^p$ is the covariate vector. Let $P_{\bfX}$ denote the marginal distribution of $\bfX$ and let $\hat P_{\bfX}$ denote its empirical distribution. We write $\hat S = S(\hat P_{\bfX})$ for a statistic or structured summary extracted from the empirical design distribution. We denote $\bfbeta \in \reals^p$ be the structural parameter of interest in a given model, with support $\cfM = \{1 \leq j \leq p: \beta^0_{j} \neq 0\}$, $s=|\cfM|$. Under a scientifically meaningful data-generating mechanism, we regard $\bfbeta_{\cfM}$ as scientifically interpretable signals, while the remaining coordinates in $\bfbeta$ are nuisance parameters. 

We assume that the outcome $Y$ depends not only on $\bfX$, but also on an unobserved latent covariate $U$, which is not uncommon. For example, the latent effect may represent cluster effect, subject-specific heterogeneity, unobserved confounding, missing data, or a combination of multiple sources. We assume that $U$ cannot be directly observed and its distribution cannot be fully specified. In such a scenario, the design information is useful when the association between $\bfX$ and $U$ can be partially characterized through the marginal distribution of $\bfX$.

To assist estimation with the design distribution, we introduce two additional quantities. First, $\bfK(\hat S)$ is a positive semidefinite matrix that determines a quadratic regularization geometry. This term carries geometric information about the distribution of predictors, by incorporating which stabilizes weakly identified directions of $\bfbeta$ and improves estimation efficiency. Second, $\bfH(\bfX; \hat S) \in \reals^q$ is a feature mapping constructed from the design distribution, which is intended to approximate the component within $U$ that is associated with $\bfX$. 

\subsection{Model and estimator} \label{subsec:model}
Consider the following latent-effect model which represents a scientifically plausible data-generating mechanism, 
\be \label{latent_model}
Y_i = m(\bfX_i; \bfbeta^0, U_i) + \vareps_i, \qquad \E(\vareps_i \mid \bfX_i, U_i)=0.
\ee
For the simplicity of exposition, one may think of a partially additive model
\be \label{latent_additive}
Y_i = \mu(\bfX_i^\top \bfbeta^0 + U_i) + \vareps_i,
\ee
where $\mu(\cdot)$ is an appropriate link or mean function. The distribution of $\bfX$ carries auxiliary information about the inference of the target parameter $\bfbeta_{\cfM}$ such that ignoring this information can result in inefficient or biased estimation of $\bfbeta_{\cfM}$. The most crucial reason for that is the association between $U_i$ and $\bfX_i$ in many cases. Suppose that, if $U_i \indep \bfX_i$, or more generally, if $\E(U_i \mid \bfX_i) = \E(U_i)$, then the latent effect does not generate design-induced target distortion. On the other hand, if $\E(U_i \mid \bfX_i) \neq \E(U_i)$, then part of the latent effect becomes predictable from $\bfX_i$, and therefore ignoring this dependence will shift the estimation target away from the underlying truth $\bfbeta^0$. This target distortion problem can be more severer in high-dimensional settings for that even population-level independence between $U_i$ and $\bfX_i$ does not preclude substantial sample-level spurious correlation. The induced target distortion can be a finite-sample phenomenon. In either case when $U_i$ is predictable from $\bfX_i$, a usual penalized estimator will use scientifically irrelevant coordinates of $\bfX_i$ to proxy for the latent effect $U_i$, thereby inducing bias in the estimated structural coefficients. 

Motivated by this observation, we assume that the component $U_i$ associated with $\bfX_i$ can be approximated through a design-assisted feature map:
\be \label{latent_proj}
U_i = \bfH(\bfX_i; \hat S)^\top \bfeta^0 + r_i,
\ee
where $\bfeta^0 \in \reals^q$ is an unknown nuisance parameter and $r_i$ is a residual latent component satisfying $\E \left\{(r_i)^2 \right\} \le \bar \delta_n^2(\hat S)$. Substituting (\ref{latent_proj}) into (\ref{latent_additive}) leads to an estimation procedure that augments the regression fit by $\bfH(\bfX_i; \hat S)^\top \bfeta$, thereby separating the structural effect $\bfbeta^0$ from the part of the latent effect that is systematically associated with $\bfX_i$. 

For a realized $\hat S$, define the corresponding plug-in population risk by
\be \label{genrisk}
R(\bfbeta, \bfeta; \hat S) = \E^\circ \left\{ \ell(\bfZ; \bfbeta, \bfeta, \hat S) \right\},
\ee
where $\E^\circ$ denotes expectation with $\hat S$ held fixed. We use $\bftheta^0 = (\bfbeta^{0,\top}, \bfeta^{0,\top})^\top$ as the oracle reference parameter. When the latent approximation is imperfect, $\bftheta^0$ need not exactly minimize $R(\cdot; \hat S)$; the resulting score discrepancy is controlled below through Condition (A5). With different loss functions, many familiar models can be viewed as special cases of (\ref{genrisk}). For example, by assuming that $Y \mid \bfX$ belongs to an exponential family with the canonical parameter $\xi = \bfX^\top \bfbeta + \bfH(\bfX; \hat S)^\top \bfeta$, then $\ell(\bfZ; \bfbeta, \bfeta, \hat S) = l\left(\bfX^\top \bfbeta + \bfH(\bfX; \hat S)^\top \bfeta; Y\right)$,
where $l(\xi; Y)$ is the negative log-likelihood function, and (\ref{genrisk}) yields a generalized linear model (GLM). In another example, given a quantile level $\tau \in (0,1)$, define $\rho_\tau(u) = u\{\tau - I(u < 0)\}$, then $\ell(\bfZ; \bfbeta, \bfeta, \hat S) = \rho_\tau \left( Y - \bfX^\top \bfbeta - \bfH(\bfX; \hat S)^\top \bfeta \right)$ yields a quantile regression model. In an example of survival analysis, let $Y = (T, \Delta)$, where $T$ is the observed failure time and $\Delta$ is the censoring indicator. Defining
\ba
\ell(\bfZ; \bfbeta, \bfeta, \hat S) = -\Delta \left[ \bfX^\top \bfbeta + \bfH(\bfX; \hat S)^\top \bfeta - \log \left\{ \sum_{j: T_j \geq T} \exp\left( \bfX_j^\top \bfbeta + \bfH(\bfX_j; \hat S)^\top \bfeta \right) \right\} \right]
\ea
yields the Cox-model partial likelihood. 
It is noteworthy that some of these special cases require non-smooth loss functions, such as quantile regression, or dependent estimating criteria, such as the Cox-model partial likelihood. Our developed theory in Section \ref{sec:thm} requires sufficiently smooth risks so that some modifications to the empirical-process and curvature conditions are needed. 

The corresponding {\it design-assisted sparse estimator} for the general framework (\ref{genrisk}) is defined as
\be \label{genest} 
\begin{aligned}
(\hat{\bfbeta},\hat{\bfeta})
= \underset{\bfbeta \in \reals^p, \bfeta \in \reals^q}{\arg\min}
\Biggl[
& \frac{1}{n}\sum_{i=1}^n \ell(\bfZ_i; \bfbeta, \bfeta, \hat S)
+ \sum_{j=1}^p p_{\lambda_1}(|\beta_j|) + \frac{\lambda_2}{2} \bfbeta^\top \bfK(\hat S)\bfbeta \\
& + \sum_{k=1}^q q_{\lambda_3} (|\eta_k|)
\Biggr],
\end{aligned}
\ee
where $p_{\lambda_1}(\cdot)$ and $q_{\lambda_3}(\cdot)$ are penalty functions. The three components of the regularization in (\ref{genest}) reflect different types of information contained in $\bfX$. By assuming sparsity on $\bfbeta^0$, the penalty on $\bfbeta$, such as LASSO or SCAD \citep{tibshirani1996regression,fan2001variable}, aims to remove irrelevant covariates and thereby achieve simultaneous variable selection and parameter estimation. The quadratic term $\bfbeta^\top \bfK(\hat S)\bfbeta$ uses information from the design distribution to stabilize estimation and reduce variation by shrinking the estimator along directions that are poorly identified or highly variable. In particular, when $\bfK(\hat S)$ is related to $\Cov(\bfX)$, it regularizes the variability of the linear predictor $\bfbeta^\top \bfX$. The augmentation term $\bfH(\bfX_i; \hat S)^\top \bfeta$ is specified to capture the part of the latent effect that is associated with $\bfX_i$, thereby reducing bias in $\hat{\bfbeta}$. This nuisance component can be estimated using regularization methods with consistent prediction of $\bfH(\bfX_i;\hat S)^\top \bfeta$.

\subsection{Special examples} \label{subsec:exp}
Our proposed framework is general, so that various modern regression problems with informative design distribution can be unified under this framework. We discuss four special examples in: clustered data analysis, causal inference, missing data, and survival analysis, with the latter two illustrated in the supplementary material. In each case, the latent effect $U_i$ represents an unmodeled source of variation, the design summary $\hat S = S(\hat P_{\bfX})$ is used to construct a feature map $H(\bfX_i; \hat S)$ that captures the part of $U_i$ associated with $\bfX_i$, and a matrix $K(\hat S)$ is used to stabilize weakly identified directions. It is noteworthy that our proposed method is generic and additional problem-specific assumptions or certain modifications are required for some of these special examples. 

{\it Clustered data analysis.}
Suppose observations are indexed by cluster $i = 1, \ldots, m$ and unit $j = 1, \ldots, n_i$, and consider the model
\be \label{clustered}
Y_{ij} = \mu\!\left( \bfX_{ij}^\top \bfbeta^0 + \bfW_{ij}^\top \bfU_i \right) + \vareps_{ij},
\qquad
\E(\vareps_{ij} \mid \bfX_{ij}, \bfW_{ij},\bfU_i) = 0,
\ee
where $\bfU_i$ is an unobserved cluster-specific latent coefficient vector. This formulation includes the random-intercept model when $\bfW_{ij} \equiv 1$, as well as other general mixed-effects models, such as the random-slope model, when $\bfW_{ij}$ contains selected covariates. If the latent cluster effects are associated with the within-cluster covariate distribution, then a natural choice of design-assisted feature map is
\ba
H(\bfX_{ij}; \hat S) = \left( \bar{\bfX}_i^\top, \bfG_i^\top \right)^\top,
\qquad
\bar{\bfX}_i = \frac{1}{n_i} \sum_{j=1}^{n_i} \bfX_{ij},
\ea
where $\bfG_i$ contains additional cluster-level summaries of the covariates, such as higher-order moments and covariance features. In this way, $H(\bfX_{ij}; \hat S)^\top \bfeta^0$ approximates the component of the cluster-specific latent effect that is associated with the observed within-cluster design information. This extends existing work of correlated-random-effects constructions, in which cluster means are used to account for dependence between regressors and latent cluster effects \citep{mundlak1978pooling}. Here, the cluster-level summaries are treated as nuisance features for separating latent heterogeneity from the structural coefficients. If the estimand itself contains a distinct between-cluster contextual effect, that effect should be parameterized separately rather than absorbed into $H$.

{\it Causal inference with hidden confounding.}
Consider an observational study with treatment $A_i$, observed confounders $\bfX_i$, and an unobserved confounder $U_i$. Suppose that $A_i$ is binary and the treatment assignment depends on both $\bfX_i$ and $U_i$. The outcome model is
\ba
Y_i = m(A_i, \bfX_i; \bfbeta^0, U_i) + \vareps_i,
\qquad
\E(\vareps_i \mid A_i, \bfX_i, U_i)=0.
\ea
Let $Y_i(a)$ denote the potential outcome under treatment level $a$. In this setting, ignorability given $\bfX_i$ alone may fail, in the sense that $Y_i(a) \nindep A_i \mid \bfX_i$, because the latent structure $U_i$ may affect both treatment assignment and outcome. If some observed variables act as proxies for the hidden confounder, then $H(\bfX_i; \hat S)$ can be constructed from these proxies or their low-dimensional summaries. For instance, if $\bfV_i \subset \bfX_i$ denotes a designated proxy subset, one
may specify 
\ba
H(\bfX_i; \hat S) = \left( \bfV_i^\top, \phi_1(\bfV_i), \ldots, \phi_q(\bfV_i) \right)^\top,
\ea
where $\phi_1, \ldots, \phi_q$ are either basis functions, factor scores, or other estimated summaries. 

When the proxy variables can be separated into treatment-inducing proxies $\bfZ_i$ and outcome-inducing proxies $\bfW_i$, the proximal-causal literature suggests an additional construction. Specifically, one may first estimate $\bfW_{c,i} = \E(\bfW_i \mid A_i, \bfX_i, \bfZ_i)$, and then include the fitted proximal-control feature $\widehat{\bfW}_{c,i}$ in the nuisance
map $H$. This construction is motivated by proximal causal learning, where $\bfW_c$ serves as a proxy-control representation of the latent confounding component. In the present framework, however, $H(\bfX_i; \hat S)^\top \bfeta$ is not interpreted as a structural part of the outcome model. Rather, it is a flexible nuisance approximation designed to reduce the component of hidden confounding that distorts estimation of the treatment effect \citep{miao2018proxy,zivich2023proximal,park2024single,tchetgen2024introduction}

\section{Theoretical properties} \label{sec:thm}
\subsection{Regularity conditions}
In this section, we study the theoretical properties of the proposed {\it design-assisted estimator} (\ref{genest}) in the general framework. All proofs of the theorems in this section are included in the supplementary material. Let
\ba
R_n(\bfbeta, \bfeta; S) = \frac{1}{n} \sum_{i=1}^n \ell(\bfZ_i; \bfbeta, \bfeta, S)
\ea
be the empirical risk. We say that the design-assisted estimator (\ref{genest}) is well specified under the following regularity conditions:
\begin{enumerate}
\item[(A1)] (Sparsity) There exist sequences $s_n$ and $t_n$ such that $\|\bfbeta^0\|_0 = s_n$, $\|\bfeta^0\|_0 = t_n$, and $(s_n + t_n) \log(p+q)=o(n)$.
\item[(A2)] (Effective latent-effect approximation)
The oracle nuisance coefficient $\bfeta^0 = \bfeta^0(\hat S)$ satisfies
\ba
\frac{1}{n} \sum_{i=1}^n \left\{ U_i - \bfH(\bfX_i; \hat S)^\top \bfeta^0 \right\}^2 = O_p\! \left(\bar\delta_n^2 (\hat S) \right).
\ea
\item[(A3)] (Design moments and tail control)
The covariate vector $\bfX$ is sub-Gaussian with mean zero and covariance matrix $\bfSigma_X$ satisfying
\ba
0 < c_X \le \lambda_{\min} (\bfSigma_X) \le \lambda_{\max}(\bfSigma_X) \le C_X < \infty.
\ea
Moreover, the coordinates of $\bfH(\bfX_i; \hat S)$ satisfy
\ba
\max_{1\le k \le q} \frac{1}{n} \sum_{i=1}^n H_k^2(\bfX_i; \hat S) = O_p(1).
\ea
\item[(A4)] (Local identification) There exists a neighborhood $\cfN_0$ of $\bftheta^0$, constants $0 < \kappa_0 \le \kappa_1 < \infty$, and $a > 0$ such that, with probability tending to one, for every $\bftheta \in \cfN_0$ and
\ba
\bfv \in \cfC(a,\cfA) = \left\{ \bfv \in \reals^{p+q}: \|\bfv_{\cfA^c}\|_1 \le a \|\bfv_{\cfA}\|_1 \right\},
\qquad
\cfA = \supp(\bftheta^0),
\ea
we have $\kappa_0 \|\bfv\|_2^2 \le \bfv^\top \nabla^2_{\bftheta} R(\bftheta; \hat S)\bfv \le \kappa_1 \|\bfv\|_2^2.$
\item[(A5)] (Score control)
For $a_{n,\beta} = \sqrt{(\log p)/n}$ and $a_{n,\eta} = \sqrt{(\log q)/n}$, we have
\ba
\left\| \nabla_{\bfbeta} R_n(\bftheta^0; \hat S) - \nabla_{\bfbeta} R(\bftheta^0; \hat S) \right\|_\infty = O_p(a_{n,\beta}), ~
\left\| \nabla_{\bfeta} R_n(\bftheta^0; \hat S) - \nabla_{\bfeta} R(\bftheta^0; \hat S) \right\|_\infty = O_p(a_{n,\eta}).
\ea
In addition, uniformly over $\bfv \in \cfC(a,\cfA)$,
\[
\left| \nabla_{\bftheta} R(\bftheta^0; \hat S)^\top \bfv \right| = O_p \! \left\{ \bar \delta_n(\hat S) \|\bfv\|_2 \right\}.
\]
\item[(A6)] (Restricted Hessian concentration) There exists a neighborhood $\cfN_0$ of $\bftheta^0$ such that
\[
\sup_{\bftheta \in \cfN_0} \| \nabla_{\bftheta}^2 R_n(\bftheta; \hat S) - \nabla_{\bftheta}^2 R(\bftheta; \hat S) \|_{op} = O_p(1).
\]
\item[(A7)] (Penalty regularity) The penalty functions $p_{\lambda_1}(\cdot)$ and $q_{\lambda_3}(\cdot)$ are nonnegative, symmetric, nondecreasing on $[0,\infty)$, and satisfy $p_{\lambda_1}(0) = q_{\lambda_3}(0) = 0$. They are continuously differentiable on $(0,\infty)$ with $p_{\lambda_1}'(0+) = \lambda_1$, $q_{\lambda_3}'(0+) = \lambda_3$, $\sup_{t>0}|p_{\lambda_1}'(t)| \le \lambda_1$, and $\sup_{t>0}|q_{\lambda_3}'(t)| \le \lambda_3$. The tuning parameters satisfy $\lambda_1 \asymp a_{n,\beta}$, $ \lambda_3 \asymp a_{n,\eta}$, and $\lambda_3 = O(\lambda_1)$.
\item[(A8)] (Quadratic design regularity) The matrix $\bfK(\hat S)$ is symmetric positive semidefinite and satisfies $\|\bfK(\hat S)\|_{op} = O_p(1)$ and $\|\lambda_2 \bfK(\hat S) \bfbeta^0\|_\infty = O(\lambda_2)$, with $\lambda_2 \go 0$ and $s_n \lambda_2 \go 0$.
\end{enumerate}
Here, Conditions (A1)–(A3) define the sparsity structure, latent-effect approximation, and design regularity. Conditions (A4)–(A8) are standard high-dimensional assumptions ensuring local identifiability, score and Hessian concentrations, penalty regularity, and stable quadratic shrinkage. Condition (A5) separates ordinary stochastic score fluctuation from the score discrepancy induced by the remaining latent approximation error. In particular, when $H(\bfX; \hat S)$ contains summaries of covariates that also enter the structural design, there must remain sufficient variation to distinguish $\bfX^\top \bfbeta$ from $H(\bfX; \hat S)^\top \bfeta$ on the relevant restricted parameter space. In some applications such as the causal inference setting, however, only a low-dimensional component of $\bfbeta$ is scientifically targeted. In such cases, full identification of every nuisance coefficient is not necessary; it is sufficient that the target direction remain identifiable after profiling out the nuisance space.

\subsection{Oracle properties}
\begin{theorem} \label{thm1}
Under (A1)-(A8), the design-assisted estimator (\ref{genest}) satisfies:
\[
\begin{aligned}
\text{(I)}\quad
&\lVert( \hat{\bfbeta} - \bfbeta^0, \hat{\bfeta} - \bfeta^0) \rVert_2^2 = O_p\! \left( (s_n + t_n) \lambda_1^2 + \lambda_2^2s_n + \bar \delta_n^2(\hat S) \right),\\[-0.25ex]
\text{(II)} \quad &\lVert \hat{\bfbeta} - \bfbeta^0 \rVert_1 = O_p\! \left( s_n (\lambda_1 + \lambda_2) + \sqrt{s_n}\, \bar \delta_n(\hat S) \right),\\[-0.25ex]
\text{(III)} \quad &R(\hat{\bfbeta}, \hat{\bfeta}; \hat S) -R(\bfbeta^0, \bfeta^0; \hat S) = O_p\! \left( (s_n + t_n) \lambda_1^2 + \lambda_2^2 s_n + \bar\delta_n^2(\hat S) \right).
\end{aligned}
\]
\end{theorem}

Theorem \ref{thm1} shows that the proposed estimator is controlled by the stochastic error, quadratic shrinkage, and latent-effect approximation. 
When the residual latent component that is not captured by $H(\bfX_i;\hat S)^\top \bfeta$ is orthogonal to the score directions of the regression model, it behaves as additional noise. In that case, the relevant target is a marginal parameter, which under suitable conditions remains approximately proportional to $\bfbeta^0$. Let
\ba
\bfbeta^{\mathrm m} = \underset{\bfbeta \in \reals^p}{\arg\min} R_{\mathrm m}(\bfbeta; \hat S),
\qquad
R_{\mathrm m}(\bfbeta; \hat S) = \E^\circ \! \left\{ \ell_m(\bfZ; \bfbeta, \hat S) \right\},
\ea
where the marginal loss is single-index in the sense that $\ell_m(\bfZ; \bfbeta, \hat S) = \tilde \ell(Y, \bfbeta^\top \bfX; \hat S)$ for some differentiable function $\tilde \ell(y, t; \hat S)$, and let $\psi(y, t; \hat S) = \frac{\partial}{\partial t} \tilde \ell(y, t; \hat S)$. The exact form of $\ell_{\mathrm m}$ is model dependent. In a linear model, exogeneity of $r_i$ yields $\bfbeta^{\mathrm m} = \bfbeta^0$ exactly. In nonlinear models, however, marginalization may instead produce a rescaled or approximately proportional target parameter.

For marginal models, we further assume the following regularity conditions:
\begin{enumerate}
\item[(A9)] $\bfX_i$ is elliptically distributed, so that the linearity condition
\ba
\E(\bfX_i \mid \bfbeta^{0,\top} \bfX_i = t) = \frac{\bfSigma_X \bfbeta^0}{\bfbeta^{0,\top} \bfSigma_X \bfbeta^0}\, t
\ea
holds, and $\bfSigma_X$ is positive definite;

\item[(A10)] There exists a sequence $\rho_n \ge 0$ and, for each scalar $c$ in a neighborhood of $c_0$, a measurable function $h_c(\cdot)$ such that
\ba
\E^\circ \! \left\{ \psi\! \left(Y_i, c\, \bfbeta^{0,\top} \bfX_i; \hat S \right) \mid \bfX_i \right\} = h_c(\bfbeta^{0,\top} \bfX_i) + e_{c,i},
\ea
where the remainder satisfies $\left\| \E^\circ (\bfX_i e_{c,i}) \right\|_2 \le \rho_n$;

\item[(A11)] 
Let $T = \bfbeta^{0,\top} \bfX_i$, there exists a scalar $c_0 \neq 0$ such that $\E^\circ \! \left\{T\, h_{c_0}(T) \right\} = 0$.
\end{enumerate}
Conditions (A9)-(A11) require that the marginal score $\psi$ can be approximately determined by a single-index direction $\bfbeta^{0, \top} \bfX_i$, up to an error of order $\rho_n$. Condition (A9) is a standard linearity condition under elliptical distributions; Condition (A10) requires the residual latent effect to preserve this approximate single-index score structure, and Condition (A11) excludes the degenerate case in which the marginal signal is zero.
\begin{proposition} \label{prop0}
Under regularity conditions (A9)-(A11), suppose that the marginal risk $R_{\mathrm m} (\bfbeta; \hat S)$ is locally strongly convex in a neighborhood of $c_0 \bfbeta^0$, then $\left\| \bfbeta^{\mathrm m} - c_0\bfbeta^0 \right\|_2 = O(\rho_n)$.
\end{proposition}

Thus, if $\rho_n=0$, then locally $\bfbeta^{\mathrm m} = c_0 \bfbeta^0$, and hence $\supp(\bfbeta^{\mathrm m}) = \supp(\bfbeta^0)$. More generally, since $\|\bfbeta^{\mathrm m} - c_0 \bfbeta^0\|_\infty = O(\rho_n)$, if $\rho_n = o(\tau_n)$, where $\tau_n = o\{ \min_{j \in \cfM} |c_0 \beta_j^0| \}$, then $\{ j: |\beta_j^{\mathrm m}| > \tau_n \} = \cfM$.

For comparison purposes, we consider the estimator $\tilde{\bfbeta}$ that uses only the sparsity-inducing penalty on $\bfbeta$ as a benchmark:
\be \label{benchest}
\tilde{\bfbeta} = \underset{\bfbeta \in \reals^p}{\arg\min} \left[ \frac{1}{n}\sum_{i=1}^n \ell(\bfZ_i;\bfbeta,\bfzero,\hat S) + \sum_{j=1}^p p_{\tilde{\lambda}_1}(|\beta_j|) \right].
\ee
Similarly to (\ref{genrisk}), the relevant target parameter is defined through
\ba
\bfbeta^\star = \underset{\bfbeta \in \reals^p}{\arg\min} R_{0}(\bfbeta; \hat S),
\qquad
R_{0}(\bfbeta; \hat S) = \E^\circ \{ \ell(\bfZ;\bfbeta, \bfzero, \hat S) \}.
\ea
To distinguish the benchmark tuning level from that of the design-assisted estimator, we use $\tilde\lambda_1$ for the sparsity penalty in (\ref{benchest}). This distinction is necessary because $\tilde{\bfbeta}$ targets $\bfbeta^\star$, which may be less sparse than $\bfbeta^0$ due to the distortion of the target. Similarly, we denote $\tilde{s}_n = \|\bfbeta^\star\|_0$.

\begin{theorem} \label{thm2}
Suppose that (A1)-(A3) and (A5)-(A7) hold for the risk $R_{0}(\bfbeta; \hat S)$. The benchmark estimator (\ref{benchest}) satisfies
\[
\begin{aligned}
\text{(I)}\quad
&\lVert \tilde{\bfbeta} - \bfbeta^\star \rVert_1
 = O_p\! \left( \tilde{s}_n \tilde{\lambda}_1 \right),\\[-0.25ex]
\text{(II)}\quad
&\lVert \tilde{\bfbeta} - \bfbeta^\star \rVert_2^2
 = O_p\! \left( \tilde{s}_n \tilde{\lambda}_1^2 \right),\\[-0.25ex]
\text{(III)}\quad
&R_0(\tilde{\bfbeta}; \hat S) - R_0(\bfbeta^\star; \hat S)
 = O_p\! \left( \tilde{s}_n \tilde{\lambda}_1^2 \right).
\end{aligned}
\]
\end{theorem}
Theorem \ref{thm2} shows the corresponding oracle rates for the benchmark estimator around its own population target $\bfbeta^\star$.

\begin{corollary} \label{cor1}
Let $\bfDelta^\star = \bfbeta^\star - \bfbeta^{\mathrm m}$ and $\bfDelta^{\mathrm m} = \bfbeta^{\mathrm m} - \bfbeta^0$. Suppose (A1)-(A11) hold for the risk $R_{0}(\bfbeta; \hat S)$. If there exists a constant $b_0 > 0$ and an integer $N_0 \ge 1$ such that
\be \label{non-negligible}
\left\| \nabla_{\bfbeta} R_0(\bfbeta^{\mathrm m}; \hat S) - \nabla_{\bfbeta} R_m(\bfbeta^{\mathrm m}; \hat S) \right\|_2 \ge b_0
\ee
for all $n > N_0$, then
\begin{enumerate}
\item[(I)] $\|\bfDelta^\star\|_2 \not\to 0$.
\item[(II)] When $\|\bfDelta^\star\|_0 = O(\tilde s_n)$, $\tilde{\bfbeta} - \bfbeta^{\mathrm m} = \bfDelta^\star + O_p(\tilde s_n \tilde \lambda_1).$
\end{enumerate}
\end{corollary}
A sufficient condition for (\ref{non-negligible}) is that $\bfeta^0 \neq \bfzero$ and the proxy component is not orthogonal to the $\bfbeta$-score direction. Corollary \ref{cor1} formalizes the benchmark target distortion: under (\ref{non-negligible}), $\tilde{\bfbeta}$ is centered at a parameter different from $\bfbeta^{\mathrm m}$ or $\bfbeta^0$.

\begin{corollary} \label{cor2}
Suppose the assumptions of Theorem \ref{thm1} hold with $\delta_n(\hat S) = 0$. Let $\kappa_0'>0$ denote the restricted curvature constant for the benchmark risk $R_0(\bfbeta; \hat S)$. Assume further that, on the restricted cone $\cfC(a, \cfM)$, we have
\ba
\bfv^\top \left\{ \nabla_{\bfbeta}^2 R(\bfbeta^0, \bfeta^0; \hat S) + \lambda_2 \bfK(\hat S) \right\} \bfv \ge \kappa_K \|\bfv\|_2^2
\ea
for some $\kappa_K>\kappa_0'$. Then
\ba
\|\hat{\bfbeta} - \bfbeta^0\|_2^2 = O_p\left( \frac{s_n (\lambda_1 + \lambda_2)^2}{\kappa_K^2} \right),
\qquad
\|\tilde{\bfbeta} - \bfbeta^\star\|_2^2 = O_p\left( \frac{\tilde{s}_n \tilde{\lambda}_1^2}{\kappa_0'^{,2}} \right).
\ea
\end{corollary}
That is, when $\lambda_2 = o(\lambda_1)$, both estimators have the same first-order rate, but since $\kappa_K > \kappa_0'$, the proposed estimator has a strictly sharper oracle bound than the benchmark estimator. A natural requirement for $K(S)$ is that it be symmetric positive semidefinite and that it contribute nontrivial curvature on the restricted cone in (A5). The simplest choice is $K(S) = \bfI_p$. In this case, for every $\bfv \in \cfC(a,\cfM)$,
\ba
\bfv^\top \left\{ \nabla_{\bfbeta}^2 R(\bfbeta^0, \bfeta^0; \hat S) + \lambda_2 \bfI_p \right\} \bfv \ge (\kappa_0 + \lambda_2) \|\bfv\|_2^2,
\ea
so that one may take $\kappa_K = \kappa_0 + \lambda_2$. Therefore, $K(S) = \bfI_p$ automatically satisfies the strengthened curvature condition in Corollary \ref{cor2}. This is the most direct analogue of ridge regression, where the penalty improves conditioning by adding uniform curvature in every direction \citep{hoerl1970ridge}. A second natural choice is $K(S) = \Cov(\bfX)$ or a regularized estimate of it. In that case, if there exists $\mu_K > 0$ such that $\bfv^\top \Cov(\bfX) \bfv \ge \mu_K \|\bfv\|_2^2$ for all $\bfv \in \cfC(a, \cfM)$, then
\ba
\bfv^\top \left\{ \nabla_{\bfbeta}^2 R(\bfbeta^0, \bfeta^0; \hat S) + \lambda_2 \Cov(\bfX) \right\} \bfv \ge (\kappa_0 + \lambda_2 \mu_K) \|\bfv\|_2^2.
\ea
Hence one may take $\kappa_K = \kappa_0 + \lambda_2 \mu_K > \kappa_0$. If $K(S)$ is constructed from the eigendecomposition of $\Cov(\bfX)$, then directions with small design curvature may be selectively regularized by assigning larger weights to the corresponding eigenspaces. In all such cases, the role of $K(S)$ is to enlarge the restricted curvature constant and therefore improve the oracle bound in Corollary \ref{cor2}. For the linear model, the quadratic term adds $\lambda_2 \bfK(\hat S)$ to the population Hessian, so directions on which $\bfK(\hat S)$ is positive acquire additional curvature. A formal decomposition is given in Supplementary Proposition S1.

The efficiency gain in Corollary \ref{cor2} is not new in itself; it is essentially the classical effect of ridge-type quadratic regularization in improving the conditioning of the risk function and stabilizing estimation along weakly identified directions \citep{hoerl1970ridge,tikhonov1977illposed,negahban2012unified}. We include this result here to emphasize that, within the present framework, such efficiency gain can be interpreted as one consequence of incorporating structured information from the marginal distribution of $\bfX$ through $K(S)$. 

Thus, relative to the benchmark estimator $\tilde{\bfbeta}$, Corollary \ref{cor1} shows that $\hat{\bfbeta}$ can reduce bias due to the target distortion, while Corollary \ref{cor2} shows that it can also reduce variance through the quadratic term $\bfbeta^\top \bfK(\hat S)\bfbeta$. Under Conditions (A9)-(A11), these improvements are achieved without altering the sparsity pattern, since $\bfbeta^{\mathrm m}$ remains proportional to $\bfbeta^0$. The preceding theorems suggest that the marginal distribution of $\bfX$ should play an explicit role in modern regression. In our framework, the design summary $S = S(P_{\bfX})$ affects estimation through two conceptually distinct objects. The matrix $K(S)$ controls the geometry of the estimation problem and appears in the quadratic penalty $\bfbeta^\top K(S)\bfbeta$, while the feature map $H(S)$ is used to learn the nuisance component of the latent effect that is associated with $\bfX$. 

To compare prediction performance, the quantity $R_0(\bfbeta^\star; \hat S) - R(\bftheta^0; \hat S)$ measures the loss of restricting the predictor to the benchmark class. When Conditions (A9)-(A11) hold, this comparison is compatible with structural support recovery because $\bfbeta^{\mathrm m}$ remains proportional to $\bfbeta^0$. If this quantity is small, then the benchmark predictor remains competitive for prediction. 

\begin{corollary} \label{cor3}
Under conditions (A1)-(A8), suppose that there exists a constant $C>0$ such that for all $\bfbeta$ in a neighborhood of $\bfbeta^0$,
\be \label{cond1}
R_0(\bfbeta; \hat S) - R(\bfbeta^0, \bfeta^0; \hat S) \le C\, \E^\circ\! \left[ \left\{ \bfH(\bfX_i; \hat S)^\top \bfeta^0 -\bfX_i^\top (\bfbeta - \bfbeta^0) \right\}^2 \right] + C\, \bar\delta_n^2(\hat S)    
\ee
and
\be \label{cond2}
\inf_{\bfbeta \in \reals^p} \E^\circ\! \left[ \left\{ \bfH(\bfX_i; \hat S)^\top \bfeta^0 - \bfX_i^\top(\bfbeta - \bfbeta^0) \right\}^2 \right] = O\!\left( \bar\delta_n^2(\hat S) \right),
\ee
we have
\ba
R_0(\tilde{\bfbeta}; \hat S) - R(\bfbeta^0, \bfeta^0; \hat S) = O\left( \bar\delta_n^2(\hat S) \right) + O_p\left( \tilde{s}_n \tilde{\lambda}_1^2 \right),
\ea
whereas
\ba
R(\hat{\bfbeta}, \hat{\bfeta}; \hat S) - R(\bfbeta^0, \bfeta^0; \hat S) = O_p\left( (s_n + t_n) \lambda_1^2 + \lambda_2^2 s_n + \bar\delta_n^2(\hat S) \right).
\ea
\end{corollary}
That is, if $s_n + t_n = O(\tilde{s}_n)$ and $\lambda_2^2 s_n = o(1)$, the proposed estimator and the benchmark estimator have the same first-order prediction risk up to stochastic estimation error. The condition (\ref{cond1}) links the benchmark risk gap to the quality of approximating the nuisance component $\bfH(\bfX_i; \hat S)^\top \bfeta^0$ by a linear combination of $\bfX_i$, while (\ref{cond2}) requires this approximation error to be no larger than order $\bar \delta_n^2(\hat S)$.

Corollary \ref{cor3} shows that $\tilde{\bfbeta}$ may remain competitive for prediction even when it is biased for $\bfbeta^0$. Its excess prediction risk is controlled by the quality with which the nuisance component $H(\bfX_i; \hat S)^\top \bfeta^0$ can be represented within the benchmark predictor class. Therefore, when the latent-effect proxy is well aligned with the linear span of the observed covariates, the benchmark oracle and the full oracle are prediction-equivalent up to a small approximation error. 

\subsection{Theoretical results on linear models}
We illustrate the oracle properties through the linear latent-effect model:
\be \label{lm}
Y_i = \bfX_i^\top \bfbeta^0 + U_i + \vareps_i,
\qquad
\E^\circ(\vareps_i \mid \bfX_i,U_i)=0,
\ee
where the corresponding design-assisted estimator is
\be \label{lm_est}
\begin{aligned}
(\hat{\bfbeta}, \hat{\bfeta}) = \underset{\bfbeta \in \reals^p, \bfeta \in \reals^q}{\arg\min} \Biggl[ & \frac{1}{2n} \|\bfY - \bfX \bfbeta - \bfH(\hat S) \bfeta\|_2^2 + \sum_{j=1}^p p_{\lambda_1}(|\beta_j|) \\
& + \frac{\lambda_2}{2} \bfbeta^\top \bfK(\hat S) \bfbeta
+ \sum_{k=1}^q q_{\lambda_3}(|\eta_k|) \Biggr].
\end{aligned}
\ee

\begin{proposition} \label{prop1}
Under the linear model (\ref{lm}), when $\E (\bfX_i \bfX_i^\top)$ is invertible, the benchmark population parameter $\bfbeta^\star$ satisfies
\ba
\bfbeta^\star = \bfbeta^0 + \left\{ \E (\bfX_i \bfX_i^\top) \right\}^{-1} \E^\circ \left[ \bfX_i U_i \right].
\ea
\end{proposition}
In the linear model, the target distortion term admits an explicit representation in Proposition \ref{prop1}. Thus, the general bias-reduction result in Corollary \ref{cor1} reduces to correcting the bias induced by $\E^\circ (\bfX_i U_i) \neq \bfzero$.

\begin{proposition} \label{prop3}
For the linear model (\ref{lm}), suppose $\E^\circ(\vareps_i \mid \bfX_i, U_i) = 0$, $\E^\circ(r_i \mid \bfX_i) = 0$, and $\E^\circ (r_i^2) \le \bar \delta_n^2(\hat S)$. Let
\ba
m_i^0 = \bfX_i^\top \bfbeta^0 + \bfH(\bfX_i; \hat S)^\top \bfeta^0,
\quad
m_i^\star = \bfX_i^\top \bfbeta^\star,
\quad
\bfbeta^\star = \arg\min_{\bfbeta \in \reals^p} \E^\circ \left\{ Y_i - \bfX_i^\top \bfbeta \right\}^2,
\ea
then
\ba
\E^\circ \left\{ (Y_i - m_i^\star)^2 \right\} = \E^\circ(\vareps_i^2) + \E^\circ(r_i^2) + \inf_{\bfbeta \in \reals^p} \E^\circ\left[ \left\{ \bfH(\bfX_i; \hat S)^\top \bfeta^0 - \bfX_i^\top (\bfbeta - \bfbeta^0) \right\}^2 \right].
\ea
If
\ba
\inf_{\bfbeta \in \reals^p} \E^\circ\left[ \left\{ \bfH(\bfX_i; \hat S)^\top \bfeta^0 - \bfX_i^\top (\bfbeta - \bfbeta^0) \right\}^2 \right] = O(\bar \delta_n^2(\hat S)),
\ea
then
\ba
\E^\circ \left\{ (Y_i - m_i^\star)^2 \right\} - \E^\circ \left\{ (Y_i - m_i^0)^2
\right\} = O(\bar \delta_n^2(\hat S)).
\ea
\end{proposition}

Proposition \ref{prop3} is a linear-model illustration of the general prediction-comparability result in Corollary \ref{cor3}. In the linear model, the benchmark target $\bfbeta^\star$ is exactly the population least-squares projection of $m_i^0$ onto the span of $\bfX_i$. Hence the prediction gap is given by the projection error of $\bfH(\bfX_i; \hat S)^\top \bfeta^0$ onto the linear span of $\bfX_i$, plus the residual approximation error.

\section{Implementation} \label{sec:implement}
The proposed estimator in (\ref{genest}) can be implemented using standard software for penalized regression once the design-assisted quantities $H(\bfX_i;\hat S)$ and $K(\hat S)$ have been constructed. Let $\bfW_i(\hat S) = \left( \bfX_i^\top,\, H(\bfX_i; \hat S)^\top \right)^\top \in \reals^{p+q}$ and denote $\bfW(\hat S)$ be the corresponding design matrix. Then (\ref{genest}) can be viewed as a sparse regularized estimator on the augmented design $\bfW(\hat S)$, with an additional quadratic penalty on the structural component $\bfbeta$. In LM or GLM, this allows the use of existing coordinate-descent or proximal-gradient algorithms for penalized $M$-estimation \citep{friedman2010regularization,yao2018efficient}.

\subsection{Choice of the design-assisted feature map}
The construction of $H(\bfX;\hat S)$ is a nuisance-modeling step and should be guided by the scientific structure of the application. When the purpose is to exploit design information alone, screening used to construct $H$ should be outcome-independent; its threshold controls the size of the nuisance dictionary rather than serving as a significance level for inference. If the resulting dictionary is very large or highly redundant, low-dimensional summaries such as principal-component, factor, or sparse-basis representations can be used, or stronger regularization can be imposed on $\bfeta$. Conditions (A2) and (A5) provide the theoretical criteria for this choice: $H$ should predict the relevant latent component well while preserving sufficient local identification of the augmented regression.

\subsection{Choice of tuning parameters}
The theoretical results provided in Section \ref{sec:thm} show that both the structural penalty on $\bfbeta$ and the nuisance penalty on $\bfeta$ are taken at the usual high-dimensional order. In practice, these parameters are usually selected from the data. Specifically, given a set of candidate values, we select the triplet that minimizes either the cross-validation error, an information criterion, or a prediction-risk criterion if a validation sample is available. This approach is the most flexible and can be implemented directly through existing packages such as \cite{liu2018data}. 

For computational simplicity, one may reduce tuning to a one-dimensional search by fixing the ratios between the $\lambda$s. A convenient choice is to let $\lambda_1 / \lambda_2$ be a fixed constant and let $\lambda_3 = \sqrt{\log(q) / \log(p)} \lambda_1$. The estimator can then be computed along a one-dimensional regularization path. Although this approach is less flexible than a full three-dimensional search, it is computationally faster and is often adequate for routine applications.

\subsection{Iterative algorithm} \label{subsec:ite}
Theorem \ref{thm1} requires the approximation error $\bar \delta_n(\hat S)$ to be small, while Corollary \ref{cor1}(II) is easiest to apply when the effective distortion $\bfDelta^\star$ has support of order $O(s_n)$. In practice, however, a high-dimensional $H(\bfX_i; \hat S)$ may substantially reduce $\bar \delta_n(\hat S)$ without yielding a sparse one-step approximation. This nonsparsity issue is most common when the dependence between $U_i$ and $\bfX_i$ is driven by spurious correlations arising from high dimensionality. Iterative procedures are therefore natural: by repeatedly updating the nuisance approximation $H(\bfX; \hat S)^\top \hat\bfeta$ and refitting the structural component, the final-stage distortion vector will achieve an effective support of order $O(s_n)$. 

Specifically, let $\hat \bftheta^{(0)}$ denote an initial estimator obtained from the proposed method (\ref{genest}). At iteration $k \ge 1$, we augment the feature mapping $H(\bfX_i; \hat S)$ with $\hat U_i^{(k)}$, denoted as $H(\bfX_i; \hat S, \hat U_i^{(k)})$, where the nuisance approximation $\hat U_i^{(k)}$ is defined as $H(\bfX; \hat S, \hat U_i^{(k-1)})^\top \hat\bfeta^{(k-1)}$. The updated estimator $\hat \bftheta^{(k)}$ is obtained by refitting (\ref{genest}) using the augmented feature $H(\bfX_i; \hat S, \hat U_i^{(k)})$. The algorithm stops when $n^{-1} \sum_{i=1}^n \left( \hat U_i^{(k)} - \hat U_i^{(k-1)} \right)^2$ is below a pre-specified value $e_{\mathrm{thr}}$.

\section{Simulation studies} \label{sec:sim}
The role of quadratic shrinkage in improving estimation efficiency has been extensively studied in the ridge and generalized-ridge literature \citep{hoerl1970ridge,tikhonov1977illposed,negahban2012unified}. Since that aspect of the proposed framework is already well understood, our simulations focus primarily on the bias reduction afforded by the design-assisted augmentation term $H(\bfX_i;\hat S)^\top\bfeta$, especially in settings where latent effects induce target distortion for sparse estimators that ignore design information. Therefore, we set $\lambda_2 = 0$ in all settings.

\subsection{Clustered data analysis} \label{subsec:sim1}
We consider the linear model (\ref{lm}) with the clustered data setting specified in (\ref{clustered}). We consider the same simulation settings as those used by \cite{ye2026synthetic}. Let $m = 400$, $n_i = 4$ for all $i = 1, \ldots, m$, and $p = 1000$ covariates, we generate $\bfX_{ij}$ from a multivariate Gaussian distribution with mean vector $\bfmu_i$ and precision matrix $\bfTheta$. We consider two settings of $\bfmu_i$. The first setting is denoted as \textit{Independent}, where $\bfmu_i$'s are heterogeneous but are not structurally related to $\alpha_i$. The second setting is denoted as \textit{Endogenous}, where $\bfmu_i$ and $\alpha_i$ are driven by the same hidden source. The detailed data-generating process, as well as the implementation details, are included in the supplementary material.

We compare four estimators: (i) the marginal-model LASSO estimator that ignores design information, denoted by $\tilde{\bfbeta}$; (ii) the proposed one-step design-assisted estimator with LASSO penalty, denoted by $\hat{\bfbeta}$; (iii) the iterative design-assisted estimator using $\lambda_{\min}$ for the nuisance update, denoted by $\hat{\bfbeta}^{I, 1}$; and (iv) the iterative design-assisted estimator using $\lambda_{\mathrm{1se}}$ for the nuisance update, denoted by $\hat{\bfbeta}^{I, 2}$.

We evaluate the estimators using the number of true positives (TP), the number of false positives (FP), the root mean squared error (RMSE) of the residuals, and $\|\hat{\bfbeta}-\bfbeta^0\|_1$, based on 100 simulated data sets. Here TP and FP assess variable selection performance, RMSE reflects prediction accuracy, and the $\ell_1$ error measures structural estimation accuracy. The results are reported in Table \ref{table_sim1}.

Across all settings, all estimators can identify all nonzero coefficients, yielding TP $= 6$. As is typical for LASSO-type procedures, $\lambda_{\mathrm{1se}}$ yields fewer FP and smaller $\ell_1$ error, whereas $\lambda_{\min}$ tends to achieve slightly smaller RMSE. In the \textit{Independent} setting, the improvements from the proposed design-assisted procedures are driven by spurious correlation. This is already visible when $p_0=0$, where there is no explicit design-induced association between the heterogeneous covariate distribution and the latent effect. Even in this case, the proposed estimators improve variable selection over the benchmark estimator $\tilde \bfbeta$, indicating that the design-assisted adjustment can remove incidental finite-sample alignment between irrelevant covariates and the latent effect.

As $p_0$ increases in the \textit{Independent} setting, the spurious-correlation mechanism becomes more severe because more covariates have heterogeneous cluster-level distributions and hence more opportunities arise for irrelevant variables to act as latent proxies. Accordingly, the advantage of incorporating design information becomes much more pronounced. At $p_0 = 500$ and $\lambda_{\mathrm{1se}}$, FP is reduced from $28.14$ for $\tilde{\bfbeta}$ to $22.65$ for $\hat{\bfbeta}$, and further to $3.92$ and $11.72$ for the iterative estimators, with a similar pattern for the $\ell_1$ error. This is consistent with the theory that repeated nuisance updates are particularly helpful when the latent distortion is non-sparse and not well captured by a single sparse approximation.

By contrast, in the \textit{Endogenous} setting, the improvement is driven by genuine dependence between the design heterogeneity and the latent effect. The one-step design-assisted estimator already captures most of the gain: at $p_0 = 500$ and $\lambda_{\mathrm{1se}}$, FP drops sharply from $68.51$ for $\tilde{\bfbeta}$ to $7.38$ for $\hat{\bfbeta}$, with only a more modest additional reduction under the iterative procedures. 

\subsection{Causal inference} \label{subsec:sim2}
We consider the causal inference setting introduced in Section \ref{subsec:exp}, with a binary treatment $A_i$, outcome $Y_i$, observed covariates $\bfX_i$, and an unobserved confounder $U_i$. Write $\bfX_i = (\bfX_{i,0}^\top, \bfZ_i^\top, \bfW_i^\top)^\top$, where $\bfX_{i,0}$ are ordinary baseline covariates, and $\bfZ_i$ and $\bfW_i$ denote treatment-inducing and outcome inducing proxy variables, respectively. The coefficient of $A_i$ is left unpenalized in all penalized regressions, and the corresponding estimate is used as the treatment-effect estimator.

We set $n = 800$ and $p = 1000$, with $\dim(\bfX_{i,0}) = 980$ and $\dim(\bfZ_i) = \dim(\bfW_i) = 10$. We generate $ U_i \sim N(0, ~1)$. Unless otherwise specified, the proxy variables are generated by
\ba
\bfZ_i = \bfGamma_Z U_i + \bfB_Z^\top \bfX_{i,0} + \bfnu_{Zi},
\qquad
\bfW_i = \bfGamma_W U_i + \bfB_W^\top \bfX_{i,0} + \bfnu_{Wi},
\ea
where $\bfnu_{Zi}$ and $\bfnu_{Wi}$ are multivariate Gaussian noise vectors with mean $\bfzero$ and covariance $0.25\,\bfI$, independently of $(\bfX_{i,0},U_i)$, $\bfB_Z = (B_{Z,lk}) \in \reals^{p_{X^0} \times q_Z}$ and $\bfB_W = (B_{W,lk}) \in \reals^{p_{X^0} \times q_W}$, with $B_{Z,lk}, B_{W,lk} \stackrel{\mathrm{iid}}{\sim} 0.8\,\delta_0 + 0.2\, N(0, ~0.1)$, and $\delta_0$ denotes a point mass at zero.

The treatment is generated from
\ba
A_i \mid (\bfX_{i,0}, \bfZ_i, U_i) \sim \mathrm{Bernoulli}\! \left[ \mathrm{expit}\! \left\{ \bfX_{i,0}^\top \bfalpha^0 + \gamma_A U_i + \delta_Z \bar Z_i \right\} \right],
\ea
and the outcome is generated from
\ba
Y_i = \tau^0 A_i + \bfX_{i,0}^\top \bfbeta^0 + \gamma_Y U_i + \delta_W \bar W_i + \vareps_i,
\qquad
\vareps_i \sim N(0, ~1),
\ea
where $\tau^0 = 1$ is the treatment effect of interest and $\bar Z_i$ and $\bar W_i$ are the sample-standardized averages of $\bfZ_i$ and $\bfW_i$, respectively. 

We consider five proxy settings. The first three settings, denoted as ``Weak proxy'', ``strong proxy'', and ``Misspecified proxy'', are designed to evaluate proxy informativeness and misspecification when $\bfZ_i$ and $\bfW_i$ are primarily noisy measurements of the latent confounder. In these settings, $\delta_Z = \delta_W = 0$, so the direct paths $\bfZ \to A$ and $\bfW \to Y$ are omitted. This corresponds to a simplified version of the proximal proxy DAG in which both proxy blocks are $U$-relevant, but the treatment-inducing and outcome-inducing roles are not enforced through direct effects. The last two settings, ``Treatment-inducing proxy'' and ``Outcome-proxy'', are designed to more closely follow the proximal negative-control structure. 

We compare two marginal-model LASSO estimators and several design-assisted estimators. The first marginal estimator uses only the baseline covariates $\bfX_{0}$ and is denoted by $\tilde{\bfbeta}^{1}$. The second uses the full observed covariate vector $(\bfX_{0}^\top, \bfZ^\top, \bfW^\top)^\top$ and is denoted by $\tilde{\bfbeta}^{2}$. 
For the design-assisted estimators, we consider two constructions of the nuisance feature map. The first construction uses the same three design-assisted estimators as in Section \ref{subsec:sim1}, while the second is motivated by proximal causal learning \citep{tchetgen2024introduction}. 

The classical proximal two-stage least-squares estimator \citep{tchetgen2024introduction} is not directly applicable in our high-dimensional setting because the number of baseline covariates is comparable to or larger than the sample size. We therefore use this LASSO-based proximal g-computation estimator as a high-dimensional benchmark that preserves the key proximal-control construction while regularizing both the first-stage proxy regression and the second-stage outcome bridge regression. We report two estimators. The first, denoted by $\hat \tau_{\mathrm{PG}}^1$, is the coefficient of $A_i$ from the second-stage LASSO fit. The second, denoted by $\hat\tau_{\mathrm{PG}}^2$, is a post-LASSO version motivated by \cite{belloni2014inference}, which is obtained by refitting ordinary least squares on $A_i$ and the nuisance variables selected in the second stage. More details are available in the supplementary material.

We evaluate the estimators using the empirical bias and RMSE for the causal parameter $\tau^0$ based on 100 simulated data sets. As expected, the $\lambda_{\mathrm{1se}}$ rule yields slightly larger bias and RMSE than $\lambda_{\min}$, so we focus mainly on the relative performance within each tuning rule. The marginal estimator $\tilde{\bfbeta}^{2}$, which includes the proxy variables directly in the working regression, improves over $\tilde{\bfbeta}^{1}$, which uses only the baseline covariates $\bfX_0$. The improvement is more pronounced in the treatment-inducing and outcome-proxy settings, where the bias of $\tilde{\bfbeta}^{1}$ is substantially larger because the proxy variables directly contribute to treatment assignment or outcome variation. In the outcome-proxy setting, the bias decreases from $0.980$ for $\tilde{\bfbeta}^{1}$ to $0.260$ for $\tilde{\bfbeta}^{2}$. This pattern is consistent with Corollary \ref{cor3}: even when the marginal working model is biased for structural estimation, including proxy variables can provide an automatic synthetic approximation to the latent confounding component and improve prediction-oriented performance.

Our proposed design-assisted estimators further reduce bias relative to the full marginal benchmark $\tilde{\bfbeta}^{2}$ in all five settings. 
Thus, explicitly incorporating the proxy-based nuisance component through $H(\bfX_i; \hat S)$ provides additional bias reduction beyond simply adding the proxy variables to the marginal model. 

The iterative estimators provide further improvement over the one-step design-assisted estimator, especially in the treatment-inducing and outcome-proxy settings. This is because these settings contain more explicit yet complex proxy-driven nuisance structure, so iterative updates can refine the one-step linear approximation of the latent confounding component. 
In the weak- and strong-proxy settings, the improvement is more modest, suggesting that the one-step adjustment already captures much of the low-dimensional proxy information.

The proximal-control design-assisted estimator $\hat{\bfbeta}^{*}$ gives mixed results. It improves over $\tilde{\bfbeta}^{2}$ in the strong- and misspecified-proxy settings, but is not uniformly superior to the raw-proxy design-assisted estimators. This is reasonable because $\hat{\bfbeta}^{*}$ uses only the fitted proximal-control component $\widehat{\bfW}_{c,i}$ together with low-dimensional summaries, whereas the raw-proxy construction keeps the full dictionary $(\bfZ_i,\bfW_i)$ available for nuisance approximation. When the fitted bridge component is imperfect, replacing the raw proxies with $\widehat{\bfW}_{c,i}$ can discard useful residual proxy information.

The LASSO-based proximal g-computation benchmarks provide a useful comparison. The post-LASSO version $\hat \tau_{\mathrm{PG}}^2$ improves on the original LASSO coefficient $\hat \tau_{\mathrm{PG}}^1$ in all settings, reflecting the reduction of shrinkage bias from least-squares refitting after variable selection. The proximal g-computation estimators perform particularly well in the strong-proxy setting, where $\hat \tau_{\mathrm{PG}}^2$ has bias $0.092$ and RMSE $0.130$. This is expected because both $\bfZ_i$ and $\bfW_i$ are strong approximately linear measurements of the same latent confounder, so $\widehat{\bfW}_{c,i}$ is an effective proxy-control representation. In contrast, the treatment-inducing and outcome-proxy settings introduce additional direct proxy paths $\bfZ \to A$ and $\bfW \to Y$, which make the data-generating mechanism closer to the qualitative proximal DAG but also increase the difficulty of the finite-sample estimation problem. The direct $\bfZ \to A$ path strengthens treatment selection, while the direct $\bfW \to Y$ path introduces outcome variation that may not be fully captured by the fitted component $\widehat{\bfW}_{c,i}$. Overall, the results show that the raw design-assisted estimators are robust across all proxy settings, while the proximal-control and proximal g-computation benchmarks are most effective when the fitted bridge approximation is strong.

\section{Real data examples} \label{sec:real}
\subsection{Clustered data example} \label{subsec:real1}
We first illustrate the proposed method using a longitudinal bulk RNA-seq data set of enriched blood neutrophils from hospitalized COVID-19 patients \citep{lasalle2022longitudinal}. The data set contains 629 blood samples from 306 patients, collected at up to three time points (day 0, day 3, and day 7). Neutrophils were enriched from whole blood for bulk RNA sequencing, and gene expression was quantified as $\log_2(\mathrm{TPM}+1)$, where TPM (transcripts per million) is a normalized measure of expression abundance. Among 5{,}000 candidate genes, we retain 3{,}481 genes with a two-sample $t$-test $p$-value below $0.05$, together with age and BMI, as candidate covariates. The outcome variable is the indicator of maximum disease severity within 28 days.

We consider the GLM model with $\xi = \mathrm{logit}[\E(Y)]$ under the clustered data setting introduced in Section \ref{subsec:exp}. We consider the same four estimators used in Section \ref{subsec:sim1}: the marginal-model estimator $\tilde{\bfbeta}$, the one-step design-assisted estimator $\hat{\bfbeta}$, and the two iterative estimators $\hat{\bfbeta}^{I,1}$ and $\hat{\bfbeta}^{I,2}$. We only consider $\lambda_{\mathrm{1se}}$, since it demonstrates better performance on variable selection in our simulation studies. For the three design-assisted estimators, we construct $H(\bfX; \hat S)$ by first carrying out an ANOVA-type patient-level screening based on the likelihood ratio test for each covariate, and then including the corresponding cluster means $\bar X_{l, i}$ for all covariates with $p$-value below $0.05$. This results in 2{,}397 cluster-level nuisance features.

The marginal-model estimator $\tilde \bfbeta$ selects 75 genes, whereas the one-step design-assisted estimator $\hat \bfbeta$ selects 37 genes. The two iterative estimators yield the same 35 selected genes, largely overlapping with the one-step design-assisted fit. The substantially larger model selected by $\tilde{\bfbeta}$ is consistent with the target-distortion phenomenon described in Corollary \ref{cor1}: when patient-level latent effects are ignored, irrelevant genes may enter the model as proxies for within-patient heterogeneity. By contrast, the design-assisted estimators produce a more parsimonious set of genes after adjusting for the component of the latent effect associated with the covariate distribution. Most of the 37 genes selected by $\hat{\bfbeta}$ are either core findings in \citet{lasalle2022longitudinal} or belong to biologically coherent pathways. Biological categories of the 37 selected genes are summarized in the supplementary material.

\subsection{Causal inference example} \label{subsec:real2}
We next reanalyze the right-heart catheterization (RHC) data from the Study to Understand Prognoses and Preferences for Outcomes and Risks of Treatments (SUPPORT) \citep{connors1996effectiveness}, following the proximal causal analysis of \citet{tchetgen2024introduction}. The scientific question is to evaluate the effect of RHC during the initial ICU care on survival time up to 30 days. The treatment variable $A_i$ indicates whether RHC was performed within the initial 24 hours of ICU stay, and the outcome $Y_i$ is the number of days between admission and death or censoring at 30 days.

After preprocessing, the analysis includes 5{,}735 patients and 42 baseline covariates. We use the ten physiological measurements collected during the first 24 hours in the ICU as candidate proxy variables: serum sodium, serum potassium, serum creatinine, bilirubin, albumin, PaO$_2$/FiO$_2$ ratio, PaCO$_2$, arterial pH, white blood cell count, and hematocrit. Following the proxy-allocation idea in \citet{tchetgen2024introduction}, we rank these candidate proxies by their adjusted associations with treatment and outcome, and then allocate them greedily into treatment-inducing proxies $\bfZ_i$ and outcome-inducing proxies $\bfW_i$. This procedure selects
\ba
\bfZ =(\texttt{pafi1}, \texttt{pot1}, \texttt{hema1}, \texttt{crea1}, \texttt{alb1}), ~ \bfW = (\texttt{bili1}, \texttt{paco21}, \texttt{ph1}, \texttt{sod1}, \texttt{wblc1}).
\ea

We compare ordinary least squares, proximal two-stage least squares (P2SLS), the two marginal LASSO estimators from the simulation study, and the proposed design-assisted estimators. The estimator $\tilde{\bfbeta}^{1}$ uses only the baseline covariates, whereas $\tilde{\bfbeta}^{2}$ additionally includes the proxy variables. The design-assisted estimators $\hat{\bfbeta}$, $\hat{\bfbeta}^{I,1}$, and $\hat{\bfbeta}^{I,2}$ use the raw proxy-based construction of $H(\bfX_i; \hat S)$. We also report the proximal-control version $\hat{\bfbeta}^{*}$, where $H(\bfX_i; \hat S)$ is constructed from $\widehat{\bfW}_{c,i} = \widehat{\E}(\bfW_i \mid A_i, \bfX_i, \bfZ_i)$, in analogy with the proximal-control variable used in P2SLS.

The ordinary least-squares estimate is $-0.486$, whereas the P2SLS estimate is $-1.066$. Thus, as in the original proximal analysis, the bridge-based estimator moves the estimated effect in a more negative direction than ordinary regression adjustment. The two marginal LASSO estimators give similar estimates, with $\tilde{\bfbeta}^{1} = -0.662$ and $\tilde{\bfbeta}^{2} = -0.675$. This suggests that directly adding the proxy variables to the marginal working model does not substantially change the estimated treatment effect after adjustment for the baseline covariates.

The raw design-assisted estimators give estimates close to the ordinary least-squares estimate: $\hat{\bfbeta} = -0.500, ~ \hat{\bfbeta}^{I,1} = -0.499, ~ \hat{\bfbeta}^{I,2} = -0.513$.
This behavior is consistent with the simulation results, where the raw proxy construction was relatively stable across settings but did not necessarily mimic the proximal bridge estimator. In the SUPPORT data, the raw proxy dictionary appears to provide a nuisance adjustment that is largely aligned with the ordinary covariate-adjusted regression. In contrast, the proximal-control design-assisted estimator gives $\hat{\bfbeta}^{*} = -0.956$. The estimator $\hat{\bfbeta}^{*}$ borrows the proximal-control idea by constructing fitted proxies $\widehat{\bfW}_c$, and therefore moves toward the P2SLS estimate.

\section{Discussion}
Motivated by regression settings with random design, high dimensionality, and latent effects, we propose the design-assisted framework in which estimation depends not only on the conditional model $Y \mid \bfX$, but also on structured features of $\bfX$ distribution. The framework unifies several existing methods as special cases and provides a general perspective on how design information may be used to improve estimation. In comparison with a benchmark sparse procedure that ignores such information, the proposed framework highlights two distinct consequences of doing so: loss of efficiency through weakly identified design directions, and target distortion through unadjusted latent-effect components associated with $\mathbf{X}$.

{\it High-dimensional nuisance learning}.
As the dimension of $\bfX$ increases, the role of $H(S)$ becomes more important. On one hand, a larger covariate space creates a richer collection of candidate summaries and proxy variables for latent effects, thereby enlarging the class of admissible nuisance representations. This makes it more plausible that $U_i$ can be approximated well by a linear combination of observed features, so that the approximation error $\bar \delta_n(\hat S)$ becomes small \citep{mulayoff2019minimal,miao2018proxy,hansen2019factor}. On the other hand, the same increase in dimensionality also makes it more likely that some observed variables align with the latent effect, even if only through weak or partly spurious correlation. As a result, the target distortion $\bfDelta^0 = \bfbeta^\star - \bfbeta^0$ will become larger. 

This dual role suggests that $H(\bfX; \hat S)^\top \bfeta$ should be viewed as a nuisance learner rather than as an interpretable low-dimensional parameter. Under this viewpoint, the goal is not about recovery of $\bfeta^0$, but to accurately predict the component of $U_i$ that distorts the $\bfbeta$-score. The iterative algorithm proposed in Section \ref{subsec:ite} is a natural extension. Further, extending the present framework to allow nonparametric nuisance representations is an important direction for future work \citep{chernozhukov2018dml,chakrabortty2018efficient}. Whether such nonparametric or machine-learning-based nuisance learners preserve the oracle properties in Theorem \ref{thm1} remains an open theoretical question. 

{\it Constructing $H(S)$ without latent effect resources}.
When no plausible latent-effect resources are available, satisfying (A2) becomes fundamentally challenging. The problem is no longer only one of regularized prediction but also one of identification, because the latent component associated with $\bfX_i$ cannot generally be recovered from observational data alone without additional structure. One possibility arises when a subset of the samples is drawn under an experimental design. In that case, instead of approximating $U_i$ directly, one can use the experimental subset to construct weights or balancing schemes that remove the contribution of $U_i$ from the target estimating equation \citep{colnet2024causal}.

{\it Standard error estimation}. 
The primary goal of this paper is to clarify how the marginal distribution of the covariates affects structural estimation and variable selection through design-induced nuisance structure. We therefore leave the standard error estimation for future work. When the contribution of the residual approximation error to the debiased score is asymptotically negligible at the $n^{-1/2}$ scale, the standard debiased LASSO procedure can be adapted to the augmented design \citep{van2014asymptotically}. In contrast, when this contribution is non-negligible, conventional debiasing does not in general remove the remaining target bias, and an additional correction for the residual latent component is required. For clustered data, such standard error estimation is developed in \citet{ye2026synthetic}. A generic inferential theory for our design-assisted $M$-estimation framework is left for future work.

\begin{table}[!h]
\centering
\caption{Simulation results under the clustered data settings.} \label{table_sim1}
\centering
\resizebox{\ifdim\width>\linewidth\linewidth\else\width\fi}{!}{
\fontsize{7}{9}\selectfont
\begin{tabular}[t]{ccccccccccccccc}
\toprule
\multicolumn{3}{c}{ } & \multicolumn{12}{c}{$p_0$} \\
\cmidrule(l{3pt}r{3pt}){4-15}
\multicolumn{3}{c}{ } & \multicolumn{3}{c}{TP} & \multicolumn{3}{c}{FP} & \multicolumn{3}{c}{RMSE} & \multicolumn{3}{c}{$\ell_1$ error} \\
\cmidrule(l{3pt}r{3pt}){4-6} \cmidrule(l{3pt}r{3pt}){7-9} \cmidrule(l{3pt}r{3pt}){10-12} \cmidrule(l{3pt}r{3pt}){13-15}
Setting & Estimator & Tuning & $0$ & $200$ & $500$ & $0$ & $200$ & $500$ & $0$ & $200$ & $500$ & $0$ & $200$ & $500$\\
\midrule
 & $\tilde{\boldsymbol\beta}$ & $\lambda_{\min}$ & 6.000 & 6.000 & 6.000 & 35.850 & 74.210 & 143.040 & 1.482 & 1.432 & 1.349 & 0.883 & 1.326 & 2.213\\
 & & $\lambda_{\mathrm{1se}}$ & 6.000 & 6.000 & 6.000 & 1.940 & 7.730 & 28.140 & 1.536 & 1.518 & 1.479 & 0.734 & 0.678 & 0.786\\
\cmidrule{2-15}
 & $\hat{\boldsymbol\beta}$ & $\lambda_{\min}$ & 6.000 & 6.000 & 6.000 & 32.710 & 63.160 & 114.900 & 1.425 & 1.356 & 1.258 & 0.835 & 1.159 & 1.802\\
 & & $\lambda_{\mathrm{1se}}$ & 6.000 & 6.000 & 6.000 & 1.910 & 5.920 & 22.650 & 1.481 & 1.442 & 1.378 & 0.727 & 0.660 & 0.742\\
\cmidrule{2-15}
 & $\hat{\boldsymbol\beta}^{I,1}$ & $\lambda_{\min}$ & 6.000 & 6.000 & 6.000 & 20.370 & 31.730 & 48.050 & 1.438 & 1.390 & 1.329 & 0.733 & 0.820 & 1.016\\
 & & $\lambda_{\mathrm{1se}}$ & 6.000 & 6.000 & 6.000 & 0.780 & 1.350 & 3.920 & 1.490 & 1.456 & 1.408 & 0.821 & 0.765 & 0.729\\
\cmidrule{2-15}
 & $\hat{\boldsymbol\beta}^{I,2}$ & $\lambda_{\min}$ & 6.000 & 6.000 & 6.000 & 19.490 & 35.400 & 71.120 & 1.437 & 1.380 & 1.292 & 0.723 & 0.856 & 1.275\\
\multirow{-8}{*}{\centering\arraybackslash Independent} & & $\lambda_{\mathrm{1se}}$ & 6.000 & 6.000 & 6.000 & 0.880 & 3.310 & 11.720 & 1.482 & 1.439 & 1.375 & 0.783 & 0.712 & 0.702\\
\cmidrule{1-15}
 & $\tilde{\boldsymbol\beta}$ & $\lambda_{\min}$ & 6.000 & 6.000 & 6.000 & 32.700 & 95.250 & 108.700 & 1.341 & 0.987 & 0.974 & 0.749 & 1.503 & 1.450\\
 & & $\lambda_{\mathrm{1se}}$ & 6.000 & 6.000 & 6.000 & 1.900 & 50.710 & 68.510 & 1.388 & 1.032 & 1.015 & 0.663 & 1.249 & 1.243\\
\cmidrule{2-15}
 & $\hat{\boldsymbol\beta}$ & $\lambda_{\min}$ & 6.000 & 6.000 & 6.000 & 31.590 & 54.310 & 48.800 & 1.296 & 0.989 & 0.985 & 0.727 & 0.810 & 0.728\\
 & & $\lambda_{\mathrm{1se}}$ & 6.000 & 6.000 & 6.000 & 1.220 & 12.070 & 7.380 & 1.347 & 1.038 & 1.033 & 0.662 & 0.600 & 0.517\\
\cmidrule{2-15}
 & $\hat{\boldsymbol\beta}^{I,1}$ & $\lambda_{\min}$ & 6.000 & 6.000 & 6.000 & 18.530 & 42.890 & 34.880 & 1.309 & 1.000 & 0.998 & 0.639 & 0.711 & 0.615\\
 & & $\lambda_{\mathrm{1se}}$ & 6.000 & 6.000 & 6.000 & 0.440 & 8.070 & 3.820 & 1.355 & 1.048 & 1.043 & 0.751 & 0.602 & 0.548\\
\cmidrule{2-15}
 & $\hat{\boldsymbol\beta}^{I,2}$ & $\lambda_{\min}$ & 6.000 & 6.000 & 6.000 & 19.660 & 38.560 & 35.460 & 1.305 & 1.001 & 0.996 & 0.639 & 0.690 & 0.628\\
\multirow{-8}{*}{\centering\arraybackslash Endogenous} & & $\lambda_{\mathrm{1se}}$ & 6.000 & 6.000 & 6.000 & 0.650 & 9.570 & 6.670 & 1.347 & 1.041 & 1.034 & 0.710 & 0.611 & 0.546\\
\bottomrule
\end{tabular}}
\end{table}

\begin{table}[!h]
\centering
\caption{\label{table_sim2}Simulation results under the causal inference settings.}
\centering
\resizebox{\ifdim\width>\linewidth\linewidth\else\width\fi}{!}{
\fontsize{7}{9}\selectfont
\begin{tabular}[t]{cccccccccccc}
\toprule
\multicolumn{2}{c}{ } & \multicolumn{10}{c}{Setting} \\
\cmidrule(l{3pt}r{3pt}){3-12}
\multicolumn{2}{c}{ } & \multicolumn{5}{c}{Bias} & \multicolumn{5}{c}{RMSE} \\
\cmidrule(l{3pt}r{3pt}){3-7} \cmidrule(l{3pt}r{3pt}){8-12}
Estimator & Tuning & Weak & Strong & Miss & Treat & Outcome & Weak & Strong & Miss & Treat & Outcome\\
\midrule
 & $\lambda_{\min}$ & 0.261 & 0.261 & 0.261 & 0.469 & 0.980 & 0.275 & 0.275 & 0.275 & 0.476 & 0.987\\

\multirow{-2}{*}{\centering\arraybackslash $\tilde{\boldsymbol\beta}^{1}$} & $\lambda_{\mathrm{1se}}$ & 0.284 & 0.284 & 0.284 & 0.486 & 0.995 & 0.297 & 0.297 & 0.297 & 0.494 & 1.003\\
\cmidrule{1-12}
 & $\lambda_{\min}$ & 0.234 & 0.157 & 0.243 & 0.258 & 0.260 & 0.248 & 0.176 & 0.257 & 0.272 & 0.274\\

\multirow{-2}{*}{\centering\arraybackslash $\tilde{\boldsymbol\beta}^{2}$} & $\lambda_{\mathrm{1se}}$ & 0.269 & 0.204 & 0.275 & 0.347 & 0.347 & 0.282 & 0.218 & 0.289 & 0.359 & 0.360\\
\cmidrule{1-12}
 & $\lambda_{\min}$ & 0.189 & 0.105 & 0.208 & 0.127 & 0.128 & 0.206 & 0.132 & 0.225 & 0.151 & 0.152\\

\multirow{-2}{*}{\centering\arraybackslash $\hat{\boldsymbol\beta}$} & $\lambda_{\mathrm{1se}}$ & 0.220 & 0.136 & 0.238 & 0.169 & 0.170 & 0.235 & 0.158 & 0.253 & 0.190 & 0.191\\
\cmidrule{1-12}
 & $\lambda_{\min}$ & 0.189 & 0.104 & 0.207 & 0.121 & 0.124 & 0.206 & 0.132 & 0.225 & 0.148 & 0.150\\

\multirow{-2}{*}{\centering\arraybackslash $\hat{\boldsymbol\beta}^{I,1}$} & $\lambda_{\mathrm{1se}}$ & 0.216 & 0.132 & 0.235 & 0.150 & 0.155 & 0.232 & 0.156 & 0.251 & 0.174 & 0.178\\
\cmidrule{1-12}
 & $\lambda_{\min}$ & 0.185 & 0.100 & 0.204 & 0.114 & 0.117 & 0.202 & 0.129 & 0.222 & 0.141 & 0.144\\

\multirow{-2}{*}{\centering\arraybackslash $\hat{\boldsymbol\beta}^{I,2}$} & $\lambda_{\mathrm{1se}}$ & 0.211 & 0.128 & 0.231 & 0.141 & 0.144 & 0.227 & 0.152 & 0.247 & 0.166 & 0.168\\
\cmidrule{1-12}
 & $\lambda_{\min}$ & 0.246 & 0.143 & 0.225 & 0.274 & 0.259 & 0.260 & 0.165 & 0.240 & 0.286 & 0.271\\

\multirow{-2}{*}{\centering\arraybackslash $\hat{\boldsymbol\beta}^{*}$} & $\lambda_{\mathrm{1se}}$ & 0.271 & 0.173 & 0.251 & 0.306 & 0.287 & 0.285 & 0.192 & 0.266 & 0.317 & 0.299\\
\cmidrule{1-12}
$\hat{\boldsymbol\tau}_{\mathrm{PG}}^1$ & -- & 0.242 & 0.135 & 0.220 & 0.254 & 0.206 & 0.256 & 0.158 & 0.236 & 0.267 & 0.221\\

$\hat{\boldsymbol\tau}_{\mathrm{PG}}^2$ & -- & 0.208 & 0.092 & 0.185 & 0.202 & 0.153 & 0.228 & 0.130 & 0.208 & 0.220 & 0.178\\
\bottomrule
\end{tabular}}
\end{table}

\clearpage
\newpage

\renewcommand{\theequation}{S\arabic{equation}}
\setcounter{equation}{0}

\renewcommand{\thefigure}{S\arabic{figure}}
\setcounter{figure}{0}

\renewcommand{\thetable}{S\arabic{table}}
\setcounter{table}{0}

\appendix
\begin{center}
{\bf\Large Supplementary Material} 
\end{center}

\section{Additional special examples}
{\it Structured latent heterogeneity.}
The setting arises when heterogeneity is indexed by latent subpopulations, hidden regimes, or shifted environments. We assume
\ba
Y_i = m(\bfX_i; \bfbeta^0, U_i) + \vareps_i,
\qquad
U_i = g(\bfX_i, \xi_i),
\qquad
\E(\vareps_i \mid \bfX_i, U_i) = 0,
\ea
where $\xi_i$ indexes latent regimes or sample-specific environments. Unlike the clustered case, the relevant structure is not tied to known group membership and must instead be learned from the covariate distribution itself. A natural choice is to let $H(\bfX_i; \hat S) = \bfG_i$, where $\bfG_i$ contains additional summaries such as estimated subgroup indicators, distances to empirical subgroup centers, principal component scores, or factor scores. For example, if $\hat{\bfc}_1, \ldots, \hat{\bfc}_G$ denote empirical subgroup centers obtained from the covariates, one may take
\ba
\bfG_i = \left( \|\bfX_i - \hat{\bfc}_1 \|_2, \ldots, \|\bfX_i - \hat{\bfc}_G \|_2 \right)^\top.
\ea
In this way, $H(\bfX_i; \hat S)^\top \bfeta^0$ approximates the component of the latent effect that varies across hidden regimes. In such settings, $K(\hat S)$ may also be chosen from the empirical covariance structure, for example, through pooled or subgroup-specific covariance estimates, to regularize weakly identified directions induced by poor overlap or changing covariate distributions. 

{\it Missing data and selection effects.}
The framework also applies to regression with incomplete outcomes or covariates, where the latent effect represents an unmodeled selection, response, or dropout mechanism. Let $R_i \in \{0, 1\}$ denote the missing data indicator. Suppose that the outcome model (1) and the observation indicator satisfy
\ba
\p(R_i = 1 \mid \bfX_i, U_i) = \pi(\bfX_i, U_i),
\ea
where $U_i$ drives informative nonresponse or dropout. Thus, even after conditioning on $\bfX_i$, the response mechanism may still depend on latent structure through $U_i$.

If the latent effect is associated with covariate summaries, one may construct $H(\bfX_i; \hat S) = \left(\psi_1 (\bfX_i), \ldots, \psi_q (\bfX_i) \right)^\top$, where $\psi_1, \ldots, \psi_q$ are prespecified or data-adaptive transformations of $\bfX_i$, such as spline basis functions, interaction terms, principal component scores, factor scores, or other low-dimensional summaries predictive of missingness. In this way, $H(\bfX_i;\hat S)^\top\bfeta$ is intended to approximate the component of the latent selection effect that is associated with $\bfX_i$. The resulting design-assisted term plays a role analogous to augmentation in semiparametric missing-data estimation, where nuisance structure is modeled only to protect the target parameter \citep{tsiatis2006semiparametric}.

{\it Survival analysis with frailty or informative censoring.}
Let $Y_i = (T_i, \Delta_i)$ denote the observed event time and censoring indicator, where $T_i = \min(T_i^\ast, C_i)$ and $\Delta_i = I(T_i^\ast \le C_i)$, with $T_i^\ast$ the latent event time and $C_i$ the censoring time, and consider the Cox model
\ba
\lambda(t \mid \bfX_i,U_i) = \lambda_0(t) \exp\!\left\{\bfX_i^\top \bfbeta^0 + U_i\right\}.
\ea
When $U_i$ is a frailty term, the construction of $H(\bfX_i; \hat S)$ is similar to the cluster data setting, where $H(\bfX_{ij}; \hat S)^\top \bfeta$ approximates the component of the frailty that is associated with the within-group covariate distribution. When $U_i$ represents a latent censoring mechanism, $H(\bfX_i;\hat S)$ may instead be constructed from covariate features that are predictive of this mechanism. For example, let $G(t \mid \bfX_i) = \p(C_i \ge t \mid \bfX_i)$ denote the conditional censoring survival function, then one may take $H(\bfX_i; \hat S) = \left(\hat G(T_i \mid \bfX_i), \bfQ_i^\top \right)^\top$, where $\bfQ_i$ contains additional summaries of covariates associated with censoring, such as basis expansions, factor scores, or subgroup summaries. 

\section{Additional theoretical properties}
\begin{proposition} \label{prop2}
For the linear model (11), suppose
\ba
\E(U_i \mid \bfX_i) = 0, 
\quad
\bfH(\bfX_i; \hat S) \equiv \bfzero,
\quad
\E(\vareps_i \mid \bfX_i,U_i)=0,
\quad
\Var(U_i + \vareps_i \mid \bfX_i)=\sigma^2.
\ea
Let
\ba
R_0(\bfbeta) = \E\left[ \frac{1}{2} \left( Y_i - \bfX_i^\top \bfbeta \right)^2 \right],
\qquad
R_{\lambda_2}(\bfbeta) = R_0(\bfbeta) + \frac{\lambda_2}{2} \bfbeta^\top \bfK(\hat S)\bfbeta,
\ea
and define
\ba
\bfSigma_X = \E(\bfX_i \bfX_i^\top),
\qquad
\bfOmega_{\lambda_2} = \bfSigma_X + \lambda_2 \bfK(\hat S),
\ea
where $\bfSigma_X$ is positive definite and $\bfK(\hat S)$ is symmetric positive semidefinite. Then the following holds:

\begin{enumerate}
\item[(I)] The minimizers of $R_0(\bfbeta)$ and $R_{\lambda_2}(\bfbeta)$ are
\ba
\bfbeta^0 =
\arg\min_{\bfbeta \in \reals^p} R_0(\bfbeta), 
\qquad
\bfbeta^{\lambda_2} = \arg\min_{\bfbeta \in \reals^p} R_{\lambda_2}(\bfbeta) = \bfOmega_{\lambda_2}^{-1}\bfSigma_X \bfbeta^0.
\ea

\item[(II)] For any $\bfbeta \in \reals^p$,
\ba
R_0(\bfbeta) - R_0(\bfbeta^0) = \frac{1}{2} (\bfbeta - \bfbeta^0)^\top \bfSigma_X (\bfbeta - \bfbeta^0),
\ea
and
\ba
R_{\lambda_2}(\bfbeta) - R_{\lambda_2}(\bfbeta^{\lambda_2}) = \frac{1}{2} (\bfbeta - \bfbeta^{\lambda_2})^\top \bfOmega_{\lambda_2} (\bfbeta - \bfbeta^{\lambda_2}).
\ea

\item[(III)] For any $\bfb \in \reals^p$,
\ba
R_{\lambda_2}(\bfbeta^{\lambda_2} + \bfb) - R_{\lambda_2} (\bfbeta^{\lambda_2}) = \frac{1}{2} \bfb^\top \bfSigma_X \bfb + \frac{\lambda_2}{2} \bfb^\top \bfK(\hat S)\bfb.
\ea
Hence
\ba
R_{\lambda_2}(\bfbeta^{\lambda_2} + \bfb) - R_{\lambda_2} (\bfbeta^{\lambda_2}) \ge R_0(\bfbeta^0+\bfb)-R_0(\bfbeta^0),
\ea
with equality if and only if $\bfb^\top \bfK(\hat S) \bfb=0$.

\item[(IV)] If there exists a subset $\cfB \subset \reals^p$ and a constant $c_K>0$ such that
\ba
\bfb^\top \bfK(\hat S) \bfb \ge c_K \|\bfb\|_2^2 \qquad \mbox{for all } \bfb \in \cfB,
\ea
then, for all $\bfb \in \cfB$,
\ba
R_{\lambda_2}(\bfbeta^{\lambda_2} + \bfb) - R_{\lambda_2}(\bfbeta^{\lambda_2}) \ge R_0(\bfbeta^0 + \bfb) - R_0(\bfbeta^0) + \frac{\lambda_2 c_K}{2}\|\bfb\|_2^2.
\ea
\end{enumerate}
\end{proposition}

Proposition \ref{prop2} makes the role of the quadratic term $\bfbeta^\top \bfK(\hat S)\bfbeta$ explicit. Relative to the unpenalized population risk, the matrix $\bfK(\hat S)$ adds curvature to the objective function through the term $\lambda_2 \bfb^\top \bfK(\hat S)\bfb$. As shown in part (III), for any perturbation direction $\bfb$, the excess penalized risk is larger than the corresponding excess unpenalized risk by exactly $\frac{\lambda_2}{2}\bfb^\top \bfK(\hat S)\bfb$. Therefore, directions that are weakly identified under $\bfSigma_X$ but strongly penalized by $\bfK(\hat S)$ become more costly to move along, which improves the conditioning of the optimization problem. Part (IV) further shows that, on directions where $\bfK(\hat S)$ is uniformly positive, this increase in curvature is at least of order $\lambda_2 |\bfb|_2^2$. Consequently, the quadratic term stabilizes estimation, reduces variability along poorly identified directions, and leads to a sharper oracle bound. In this sense, the efficiency gain comes from strengthening the local geometry of the risk function rather than from changing the structural target itself. This interpretation is closely related to the classical role of ridge regression \citep{hoerl1970ridge} and Tikhonov regularization in improving the conditioning of ill-posed problems \citep{tikhonov1977illposed}, and to the oracle-inequality literature, in which stronger restricted curvature yields sharper error bounds for regularized estimators \citep{negahban2012unified}.

\section{Proofs}
\subsection{Proof of Theorem 1}
Let
\be \label{eq:proof_Qn}
\begin{split}
Q_n(\bfbeta, \bfeta; \hat S) &= \frac{1}{n}\sum_{i=1}^n \ell(\bfZ_i; \bfbeta, \bfeta, \hat S) + \sum_{j=1}^p p_{\lambda_1}(|\beta_j|) + \frac{\lambda_2}{2}\bfbeta^\top \bfK(\hat S)\bfbeta \\ &\quad + \sum_{k=1}^q q_{\lambda_3}(|\eta_k|).    
\end{split}
\ee
Since $(\hat{\bfbeta}, \hat{\bfeta})$ minimizes $Q_n(\bfbeta,\bfeta;\hat S)$,
\be \label{eq:basicineq}
Q_n(\hat{\bfbeta},\hat{\bfeta};\hat S) \le Q_n(\bfbeta^0,\bfeta^0;\hat S).
\ee

Define
\ba
\bfDelta_\beta = \hat{\bfbeta} - \bfbeta^0, \qquad \bfDelta_\eta = \hat{\bfeta} - \bfeta^0, \qquad \bfDelta =(\bfDelta_\beta^\top, \bfDelta_\eta^\top)^\top.
\ea
We first compare the empirical loss terms in \eqref{eq:basicineq}. By Taylor expansion, at some intermediate point $\bar \bftheta$ on the line segment joining $\bftheta^0$ and $\hat{\bftheta}$,
\ba
R_n(\hat \bftheta; \hat S) - R_n(\bftheta^0; \hat S) = \nabla_{\bftheta} R_n(\bftheta^0, \hat S)^\top \bfDelta + \frac{1}{2} \bfDelta^\top \nabla_{\bftheta}^2 R_n(\bar \bftheta, \hat S)\bfDelta.
\ea
By (A4) and (A6), uniformly on the restricted cone,
\ba
\bfDelta^\top \nabla_{\bftheta}^2 R_n(\bar \bftheta; \hat S) \bfDelta \ge \frac{\kappa_0}{2} \|\bfDelta\|_2^2
\ea
with probability tending to one.

For the score term, write
\[
\nabla R_n(\bftheta^0; \hat S) = \{ \nabla R_n(\bftheta^0; \hat S) -\nabla R(\bftheta^0; \hat S)\} + \nabla R(\bftheta^0; \hat S).
\]
By Condition (A5),
\[
\left| \nabla R_n(\bftheta^0; \hat S)^\top \bfDelta \right| \le C_1 a_{n,\beta} \|\bfDelta_\beta\|_1 + C_2 a_{n,\eta} \|\bfDelta_\eta\|_1 + C_3 \bar \delta_n(\hat S) \|\bfDelta\|_2
\]
with probability tending to one.

We now control the penalty terms. From \eqref{eq:basicineq},
\ba
& \frac{1}{n} \sum_{i=1}^n \Bigl\{ \ell(\bfZ_i; \hat{\bftheta}, \hat S) - \ell(\bfZ_i; \bftheta^0, \hat S) \Bigr\} + \sum_{j=1}^p \Bigl\{ p_{\lambda_1}(|\hat\beta_j|) -p_{\lambda_1}(|\beta_j^0|) \Bigr\} \\
&\qquad + \sum_{k=1}^q \Bigl\{ q_{\lambda_3}(|\hat\eta_k|) - q_{\lambda_3}(|\eta_k^0|) \Bigr\} + \frac{\lambda_2}{2} \Bigl( \hat{\bfbeta}^\top \bfK(\hat S) \hat{\bfbeta} - (\bfbeta^0)^\top \bfK(\hat S)\bfbeta^0 \Bigr) \le 0.
\ea
Using (A7) and standard decomposability arguments for sparse penalties, we obtain the cone condition
\be \label{eq:cone}
\|\bfDelta_{\cfA^c}\|_1 \le a \|\bfDelta_{\cfA}\|_1
\ee
with probability tending to one. Consequently,
\ba
\|\bfDelta\|_1 \le (1 + a) \|\bfDelta_{\cfA}\|_1 \le (1+a) \sqrt{s_n + t_n} \|\bfDelta\|_2.
\ea

For the quadratic term,
\ba
\frac{\lambda_2}{2} \Bigl( \hat{\bfbeta}^\top \bfK(\hat S) \hat{\bfbeta} - (\bfbeta^0)^\top \bfK(\hat S) \bfbeta^0 \Bigr) = \lambda_2 (\bfbeta^0)^\top \bfK(\hat S)\bfDelta_\beta + \frac{\lambda_2}{2} \bfDelta_\beta^\top \bfK(\hat S) \bfDelta_\beta.
\ea
By (A8),
\ba
\left| \lambda_2 (\bfbeta^0)^\top \bfK(\hat S) \bfDelta_\beta \right| = o_p(\lambda_2 \|\bfDelta_\beta\|_1),
\ea
and
\[
0 \le \frac{\lambda_2}{2} \bfDelta_\beta^\top \bfK(\hat S) \bfDelta_\beta = O_p(\lambda_2) \|\bfDelta_\beta\|_2^2.
\]

Finally, under (A2), the latent-effect approximation error contributes a perturbation of order $\bar\delta_n(\hat S)$ to the score, so the basic inequality yields
\[
\begin{split}
\frac{\kappa_0}{4} \|\bfDelta\|_2^2 \le{}& C \lambda_1 \sqrt{s_n} \|\bfDelta\|_2 + C \lambda_3 \sqrt{t_n} \|\bfDelta\|_2 \\
&+ C \lambda_2 \sqrt{s_n} \|\bfDelta\|_2 + C \bar \delta_n(\hat S) \|\bfDelta\|_2.
\end{split}
\]
Therefore,
\[
\|\bfDelta\|_2 = O_p\! \left\{ \sqrt{s_n} \lambda_1 + \sqrt{t_n} \lambda_3 + \sqrt{s_n} \lambda_2 + \bar \delta_n(\hat S) \right\},
\]
and hence
\[
\|\bfDelta\|_2^2 = O_p \!\left\{ s_n \lambda_1^2 + t_n \lambda_3^2 + s_n \lambda_2^2 + \bar \delta_n^2(\hat S) \right\}.
\]
Since $\lambda_3 = O(\lambda_1)$, this proves part (I).

Note that
\ba
\|\hat{\bfbeta} - \bfbeta^0\|_1 \le \sqrt{s_n}\|\hat{\bfbeta} - \bfbeta^0\|_2 + \|\bfDelta_{\beta, \cfM^c}\|_1.
\ea
By the cone condition, $\|\bfDelta_{\beta, \cfM^c}\|_1 \le a \|\bfDelta_{\beta, \cfM}\|_1$, hence
\ba
\|\hat{\bfbeta} - \bfbeta^0\|_1 = O_p \left( s_n \lambda_1 + s_n \lambda_2 + \sqrt{s_n} \bar\delta_n(\hat S) \right),
\ea
and part (II) is proved.

Finally, by a Taylor expansion of $R(\bftheta; \hat S)$,
\[
R(\hat \bftheta; \hat S) - R(\bftheta^0; \hat S) = \nabla R(\bftheta^0; \hat S)^\top \bfDelta + \frac{1}{2} \bfDelta^\top \nabla^2 R(\bar \bftheta; \hat S) \bfDelta.
\]
Conditions (A4) and (A5) imply
\[
\left| R(\hat \bftheta; \hat S) - R(\bftheta^0; \hat S) \right| = O_p\! \left\{ \bar \delta_n(\hat S) \|\bfDelta\|_2 + \|\bfDelta\|_2^2 \right\},
\]
and Part (I), together with Young's inequality, gives Part (III).

\subsection{Proof of Proposition 1}
Since
\ba
\nabla_{\bfbeta} R_{\mathrm m}(\bfbeta; \hat S) = \E \!\left[ \bfX\, \psi(Y, \bfbeta^\top \bfX; \hat S)
\right].
\ea
evaluating $\nabla_{\bfbeta} R_{\mathrm m}(\bfbeta; \hat S)$ at $\bfbeta = c_0\bfbeta^0$ and using (A10) yields
\ba
\nabla_{\bfbeta} R_{\mathrm m} (c_0 \bfbeta^0; \hat S) = \E^\circ\! \left[ \bfX_i\, \left\{ h_{c_0}(\bfbeta^{0,\top} \bfX_i) + e_{c_0,i} \right\} \right] = \E^\circ\! \left[ \bfX_i\, h_{c_0}(T) \right] + \E^\circ(\bfX_i e_{c_0,i}).
\ea
By (A9),
\ba
\E(\bfX_i \mid T) = \frac{\bfSigma_X \bfbeta^0}{\bfbeta^{0,\top} \bfSigma_X \bfbeta^0}\,T,
\ea
and therefore
\ba
\E^\circ\! \left[ \bfX_i\, h_{c_0}(T) \right] = \E^\circ\!\left[ \E(\bfX_i \mid T)\, h_{c_0}(T) \right] = \frac{\bfSigma_X \bfbeta^0}{\bfbeta^{0,\top} \bfSigma_X \bfbeta^0} \E^\circ\! \left\{ T\, h_{c_0}(T) \right\}.
\ea
By (A11), the last expectation is zero, so
\ba
\left\| \nabla_{\bfbeta} R_{\mathrm m}(c_0 \bfbeta^0; \hat S) \right\|_2 = \left\| \E^\circ(\bfX_i r_{c_0,i}) \right\|_2 = O(\rho_n).
\ea
The local strong convexity then yields
\ba
\left\| \bfbeta^{\mathrm m} - c_0\bfbeta^0 \right\|_2 = O(\rho_n).
\ea
The support conclusion follows from the beta-min condition.

\subsection{Proof of Theorem 2}
Define the benchmark objective
\ba
Q_{0,n}(\bfbeta; \hat S) = \frac{1}{n} \sum_{i=1}^n \ell(\bfZ_i; \bfbeta, \bfzero, \hat S) + \sum_{j=1}^p p_{\tilde \lambda_1}(|\beta_j|).
\ea
Since $\tilde{\bfbeta}$ minimizes $Q_{0,n}(\bfbeta; \hat S)$,
\ba
Q_{0,n}(\tilde{\bfbeta}; \hat S) \le Q_{0,n}(\bfbeta^\star; \hat S).
\ea
Repeating the argument in the proof of Theorem 1, but without the $\bfeta$-component and the quadratic term, yields $\|\bfDelta^\star\|_1 = O_p(\tilde{s}_n \tilde{\lambda}_1)$, $\|\bfDelta^\star\|_2^2 = O_p\left(\tilde{s}_n \tilde{\lambda}_1^2 \right)$, and $R_0(\tilde{\bfbeta}; \hat S) - R_0(\bfbeta^\star; \hat S) = O_p \left( \tilde{s}_n \tilde{\lambda}_1^2 \right)$.

\subsection{Proof of Corollary 1}
By definition,
\[
\nabla_{\bfbeta}R_{\mathrm m}(\bfbeta^{\mathrm m}; \hat S)=0.
\]
Hence (8) implies
\[
\|\nabla_{\bfbeta} R_0(\bfbeta^{\mathrm m}; \hat S)\|_2\ge b_0.
\]
On the other hand,
\[
\nabla_{\bfbeta} R_0(\bfbeta^\star; \hat S)=0.
\]
By Theorem 2, for some point $\bar\bfbeta$ on the line segment joining $\bfbeta^{\mathrm m}$ and $\bfbeta^\star$,
\[
\nabla_{\bfbeta}R_0 (\bfbeta^{\mathrm m}; \hat S) = \nabla_{\bfbeta}^2 R_0(\bar\bfbeta; \hat S) (\bfbeta^{\mathrm m} - \bfbeta^\star).
\]
Since the Hessian is locally bounded in operator norm by $C_0$, we have
\[
b_0 \le \|\nabla_{\bfbeta} R_0(\bfbeta^{\mathrm m}; \hat S)\|_2 \le C_0 \|\bfbeta^\star - \bfbeta^{\mathrm m}\|_2.
\]
Therefore
\[
\|\bfDelta^\star\|_2 \ge b_0/C_0,
\]
which proves (I).

For (II), since
\[
\tilde{\bfbeta} - \bfbeta^{\mathrm m} = (\tilde{\bfbeta} - \bfbeta^\star) + (\bfbeta^\star - \bfbeta^{\mathrm m}) = (\tilde{\bfbeta} - \bfbeta^\star) + \bfDelta^\star,
\]
Theorem 2 gives
\[
\|\tilde{\bfbeta} - \bfbeta^\star\|_1 = O_p(\tilde s_n \tilde \lambda_1).
\]
Hence,
\[
\|\tilde{\bfbeta} - \bfbeta^{\mathrm m} - \bfDelta^\star\|_1 = O_p(\tilde s_n \tilde \lambda_1).
\]

\subsection{Proof of Corollary 2}
When $\bar \delta_n(\hat S) = 0$, the basic inequality in the proof of Theorem 1 becomes 
\be \label{basic2}
\frac{\kappa_K}{2} \|\bfDelta_\beta\|_2^2 \le C_1 \lambda_1 \|\bfDelta_\beta\|_1 + C_2 \lambda_2 \|\bfDelta_\beta\|_1.
\ee
Due to the standard cone condition,
\ba
\|\bfDelta_{\beta, \cfM^c}\|_1 \le a \|\bfDelta_{\beta, \cfM}\|_1,
\ea
we have
\be \label{e1}
\|\bfDelta_\beta\|_1 \le (1+a) \|\bfDelta_{\beta, \cfM}\|_1 \le (1+a) \sqrt{s_n} \| \bfDelta_\beta \|_2.
\ee
Substituting (\ref{e1}) into (\ref{basic2}) gives
\ba
\frac{\kappa_K}{2} \|\bfDelta_\beta\|_2^2 \le C \sqrt{s_n} (\lambda_1 + \lambda_2) \|\bfDelta_\beta\|_2
\ea
for some constant $C>0$. Therefore,
\ba
\|\hat{\bfbeta} - \bfbeta^0\|_2 = O_p \left( \frac{\sqrt{s_n}(\lambda_1 + \lambda_2)}{\kappa_K} \right),
\ea
and hence
\ba
\|\hat{\bfbeta} - \bfbeta^0\|_2^2
= O_p \left( \frac{s_n(\lambda_1 + \lambda_2)^2}{\kappa_K^2} \right).
\ea

For $\tilde{\bfbeta}$, the basic inequality in the proof of Lemma 1 becomes
\ba
\frac{\kappa_0'}{2} \|\bfDelta^\star\|_2^2 \le C_0 \tilde{\lambda}_1 \|\bfDelta^\star\|_1.
\ea
Again, using the cone condition,
\ba
\|\bfDelta^\star\|_1 \le (1+a) \sqrt{\tilde{s}_n} \|\bfDelta^\star\|_2,
\ea
we obtain
\ba
\frac{\kappa_0'}{2} \|\bfDelta^\star\|_2^2 \le C \sqrt{\tilde{s}_n} \tilde{\lambda}_1 \|\bfDelta^\star\|_2,
\ea
and therefore
\ba
\|\tilde{\bfbeta} - \bfbeta^\star\|_2 = O_p \left( \frac{\sqrt{\tilde{s}_n} \tilde{\lambda}_1}{\kappa_0'}
\right),
\ea
which implies
\ba
\|\tilde{\bfbeta} - \bfbeta^\star\|_2^2 = O_p \left( \frac{\tilde{s}_n \tilde{\lambda}_1^2}{\kappa_0'^2}
\right).
\ea

\subsection{Proof of Corollary 3}
Since $\bfbeta^\star$ minimizes $R_0(\bfbeta; \hat S)$,
\ba
R_0(\bfbeta^\star; \hat S) - R(\bfbeta^0, \bfeta^0; \hat S) \le R_0(\bar{\bfbeta}; \hat S) - R(\bfbeta^0, \bfeta^0; \hat S)
\ea
for every $\bar{\bfbeta} \in \reals^p$. By the assumed risk domination inequality,
\ba
R_0(\bar{\bfbeta}; \hat S) - R(\bfbeta^0, \bfeta^0; \hat S) \le C \,\E^\circ \!\left[ \left\{ \bfH(\bfX_i; \hat S)^\top \bfeta^0 - \bfX_i^\top (\bar{\bfbeta} - \bfbeta^0) \right\}^2 \right] + C\, \bar\delta^2_n(\hat S).
\ea
Taking the infimum over $\bar{\bfbeta} \in \reals^p$ yields
\ba
R_0(\bfbeta^\star; \hat S) - R(\bfbeta^0, \bfeta^0; \hat S) &\le& C \inf_{\bfbeta \in \reals^p} \E^\circ \!\left[ \left\{ \bfH(\bfX_i; \hat S)^\top \bfeta^0 - \bfX_i^\top (\bfbeta - \bfbeta^0) \right\}^2 \right] + C\,\bar\delta^2_n(\hat S) \\
&=& O\left( \bar\delta^2_n(\hat S) \right).
\ea

By Lemma 1,
\ba
R_0(\tilde{\bfbeta}; \hat S) - R(\bfbeta^0, \bfeta^0; \hat S) &=& \Bigl\{ R_0(\tilde{\bfbeta}; \hat S) - R_0(\bfbeta^\star;\hat S) \Bigr\} + \Bigl\{ R_0(\bfbeta^\star; \hat S) - R(\bfbeta^0, \bfeta^0; \hat S) \Bigr\} \\ 
&=& O_p \left( \tilde s_n \tilde \lambda_1^2 \right) + O\left( \bar\delta^2_n(\hat S) \right),
\ea
which proves the first equality. The second equality can be proved analogously. 

\subsection{Proof of Proposition 2}
Under the linear model,
\ba
R_0 (\bfbeta; \hat S) = \E^\circ \left[ \frac{1}{2} \left( Y_i-\bfX_i^\top \bfbeta \right)^2 \right].
\ea
Differentiating with respect to $\bfbeta$,
\ba
\nabla_{\bfbeta} R_0(\bfbeta; \hat S) &=& -\E^\circ\left[ \bfX_i \left( Y_i-\bfX_i^\top \bfbeta \right)\right] \\
&=& -\E^\circ \left[ \bfX_i \left( \bfX_i^\top \bfbeta^0 + U_i + \vareps_i - \bfX_i^\top \bfbeta \right) \right].
\ea
At $\bfbeta = \bfbeta^\star$, the first-order condition gives
\ba
\E(\bfX_i \bfX_i^\top) (\bfbeta^\star - \bfbeta^0) = \E^\circ(\bfX_i U_i) + \E^\circ(\bfX_i\vareps_i).
\ea
Since $\E^\circ (\vareps_i \mid \bfX_i, U_i) = 0$, we have $\E^\circ (\bfX_i\vareps_i) = \bfzero$. Therefore, if $\E (\bfX_i\bfX_i^\top)$ is invertible,
\ba
\bfbeta^\star = \bfbeta^0 + \left\{ \E (\bfX_i\bfX_i^\top) \right\}^{-1} \E^\circ (\bfX_i U_i).
\ea

\subsection{Proof of Proposition 3}
Under $\E^\circ (U_i \mid \bfX_i) = 0$ and $\E^\circ (\vareps_i \mid \bfX_i, U_i)=0$, we have $\E^\circ (Y_i \mid \bfX_i) = \bfX_i^\top \bfbeta^0$. Hence
\ba
R_0(\bfbeta) = \E\left[ \frac{1}{2} \left( Y_i - \bfX_i^\top \bfbeta \right)^2 \right] = R_0(\bfbeta^0) + \frac{1}{2} (\bfbeta - \bfbeta^0)^\top \bfSigma_X (\bfbeta - \bfbeta^0),
\ea
which proves the first identity in part (II), and implies that
$\bfbeta^0=\arg\min_{\bfbeta} R_0(\bfbeta)$.

Next, because
\ba
R_{\lambda_2}(\bfbeta) = R_0(\bfbeta) + \frac{\lambda_2}{2} \bfbeta^\top \bfK(\hat S) \bfbeta,
\ea
we have
\ba
\nabla R_{\lambda_2}(\bfbeta) = \bfSigma_X (\bfbeta - \bfbeta^0) + \lambda_2 \bfK(\hat S) \bfbeta.
\ea
Setting this equal to $\bfzero$ gives
\ba
\left\{ \bfSigma_X + \lambda_2 \bfK(\hat S) \right\} \bfbeta^{\lambda_2} = \bfSigma_X \bfbeta^0,
\ea
which implies
\ba
\bfbeta^{\lambda_2} = \bfOmega_{\lambda_2}^{-1} \bfSigma_X  \bfbeta^0.
\ea
This proves part (I).

Using the quadratic expansion around $\bfbeta^{\lambda_2}$,
\ba
R_{\lambda_2} (\bfbeta) - R_{\lambda_2} (\bfbeta^{\lambda_2}) = \frac{1}{2} (\bfbeta - \bfbeta^{\lambda_2})^\top \bfOmega_{\lambda_2} (\bfbeta - \bfbeta^{\lambda_2}),
\ea
which proves the second identity in part (II).

Now set $\bfbeta = \bfbeta^{\lambda_2} + \bfb$. Then
\ba
R_{\lambda_2} (\bfbeta^{\lambda_2} + \bfb) - R_{\lambda_2}(\bfbeta^{\lambda_2}) = \frac{1}{2} \bfb^\top \bfOmega_{\lambda_2} \bfb = \frac{1}{2} \bfb^\top \bfSigma_X \bfb + \frac{\lambda_2}{2} \bfb^\top \bfK(\hat S)\bfb.
\ea
Since
\ba
R_0(\bfbeta^0 + \bfb) - R_0(\bfbeta^0) = \frac{1}{2} \bfb^\top \bfSigma_X \bfb,
\ea
part (III) follows immediately. Part (IV) follows by substituting the lower bound $\bfb^\top \bfK(\hat S) \bfb \ge c_K \|\bfb\|_2^2$ into (III).

\subsection{Proof of Proposition 4}
Write
\ba
Y_i = \bfX_i^\top \bfbeta^0 + \bfH(\bfX_i; \hat S)^\top \bfeta^0 + r_i + \vareps_i = m_i^0 + r_i + \vareps_i.
\ea
Since $\E^\circ (r_i + \vareps_i \mid \bfX_i) = 0$, for any $\bfbeta \in \reals^p$,
\ba
\E\left\{ (Y_i - \bfX_i^\top \bfbeta)^2 \right\} &=& \E^\circ \left[ \left\{ \bfH(\bfX_i; \hat S)^\top \bfeta^0 - \bfX_i^\top (\bfbeta - \bfbeta^0) \right\}^2 \right] + \E \left\{ (r_i + \vareps_i)^2 \right\}.
\ea
Minimizing over $\bfbeta$ yields
\be \label{lm_pred1}
\E^\circ \left\{ (Y_i - m_i^\star)^2 \right\} &=& \inf_{\bfbeta \in \reals^p} \E^\circ \left[ \left\{ \bfH(\bfX_i; \hat S)^\top \bfeta^0 - \bfX_i^\top (\bfbeta - \bfbeta^0) \right\}^2 \right] \nonumber \\
&~&+ \E\left\{ (r_i + \vareps_i)^2 \right\}.
\ee
On the other hand,
\be \label{lm_pred2}
\E^\circ \left\{ (Y_i-m_i^0)^2 \right\} = \E^\circ \left\{ (r_i + \vareps_i)^2
\right\}.
\ee
Subtracting (\ref{lm_pred1}) and (\ref{lm_pred2}) proves the proposition.

\section{Simulation settings}
\subsection{Clustered data analysis}
Let $m = 400$, $n_i = 4$ for all $i = 1, \ldots, m$, and $p = 1000$ covariates, we generate $\bfX_{ij}$ from a multivariate Gaussian distribution with mean vector $\bfmu_i$ and precision matrix $\bfTheta$, where $\bfTheta$ is a block diagonal matrix with block size $5$. Within each block, we set $\Theta_{ll} = 1$ and $\Theta_{ll'} = 0.5$ for $l \ne l'$. For each replication, we randomly select a set of $p_0 = 0$, $200$, or $500$ covariates, denote the corresponding index set as $\cfP_0$, to have heterogeneous distributions across clusters. Specifically, for $l \in \cfP_0$, we independently sample $\mu_{l,i}$ from a zero mean-distribution $F_l$, and for $l \in \cfP_0^c$, we set $\mu_{l,i} = 0$ for all $i$. The latent cluster effects $U_i$ are randomly sampled from a zero-mean distribution $F_U$. The active set is $\cfM = \{ 1,6,11,12,16,17 \}$, with $\bfbeta_{\cfM} = (0.5, 0.5, 1, 1, 1.5, 1.5)^\top$.

We consider two settings of $F_U$ and $F_l$. The first setting, denoted as \textit{Independent}, generates $\mu_{l,i}$ and $U_i$ independently for all $l \in \cfP_0$, with $\mu_{l,i} \sim N(0, ~1)$ and $U_i \sim 0.5N(-1, ~0.5) + 0.5N(1, ~0.5)$. This is the setting for the effect of high-dimensional spurious correlation, where, although the cluster means are heterogeneous, they are not structurally related to the latent cluster effect. The second setting, denoted as \textit{Endogenous}, first generates a hidden variable $V_i$ from $0.5N(-1, ~0.5) + 0.5N(1, ~0.5)$, and subsequently generates $\mu_{l,i}$ and $U_i$ from
\ba
\mu_{l,i} = h_l V_i + n_i^{-1} Z_{l,i}, ~~ U_i = 0.8 V_i + 0.2 Z_{i}, 
\ea
where $Z_{l, i}$ and $Z_i$ follow $N(0, ~1)$ and $h_l \sim \text{U}(0,1)$. In this case, the design heterogeneity and the latent cluster effect are driven by the same hidden source, so $\tilde \bfbeta$ is expected to suffer genuine target distortion.

For $\hat{\bfbeta}$ and $\hat{\bfbeta}^{I, \cdot}$, we construct $H(\bfX_i; \hat S)$ by first screening each covariate with a one-way ANOVA across clusters, and then including the corresponding cluster means $\bar X_{l,i}$ for all $l$ with ANOVA $p$-value below $0.05$. 

All tuning parameters are selected by ten-fold cross-validation (CV), using either the minimizer of the CV error, denoted by $\lambda_{\min}$, or the one-standard-error rule, denoted by $\lambda_{\mathrm{1se}}$ \citep{hastie2009elements}. For computational simplicity, we set $\lambda_3 = \sqrt{\frac{\log q}{\log p}}\, \lambda_1$ for the design-assisted estimators. For the iterative estimators, we set the stopping threshold to $e_{\mathrm{thr}}=0.01$.

\subsection{Causal inference}
The baseline covariates $\bfX_{i,0}$ are generated from a multivariate Gaussian distribution with mean vector $\bfmu_{X,i}$ and precision matrix $\bfTheta$, where $\bfTheta$ is the same block-diagonal matrix used in the clustered-data setting. As presented in proximal causal DAGs \citep{tchetgen2024introduction}, we allow weak dependence between $U_i$ and $\bfX_{i,0}$ by letting $\bfmu_{X,i} = \rho_X\bfB_X U_i$, where $\rho_X = 0.1$ and $\bfB_X$ is a randomly selected sparse loading vector of size 20. We set $\cfM_\alpha = \{1, 6, 11, 16\}$, $\bfalpha^0_{\cfM_\alpha} = (0.3, 0.3, 0.6, 0.6)^\top$, $\cfM_\beta = \{2, 6, 12, 13, 16, 17\}$, and $\bfbeta^0_{\cfM_\beta} = (0.5, 0.5, 1, 1, 1.5, 1.5)^\top$,  with all remaining coordinates equal to zero, and take $\tau^0=1$ and $\gamma_A = \gamma_Y = 0.5$.

In the ``Weak proxy'' setting, the observed proxy variables contain only limited information about the latent confounder, with $\bfGamma_Z = 0.4 * \bfone$ and $\bfGamma_W = 0.2 * \bfone$. In the ``Strong proxy'' setting, the proxy variables become more informative about the hidden confounder, with $\bfGamma_Z = 0.8 * \bfone$ and $\bfGamma_W = 0.6 * \bfone$. In the ``Misspecified proxy'' setting, the proxies depend on $U_i$ nonlinearly:
\ba
\bfZ_i = \bfGamma_Z (U_i^2 - 1) + \bfB_Z^\top \bfX_{i,0} + \bfnu_{Zi},
\qquad
\bfW_i = \bfGamma_W \sin(U_i) + \bfB_W^\top \bfX_{i,0} + \bfnu_{Wi},
\ea
with $\bfGamma_Z = 0.8 * \bfone$ and $\bfGamma_W = 0.6 * \bfone$. This setting evaluates robustness when a linear nuisance representation cannot fully capture the proxy-confounder relationship.

The proxies in `Treatment-inducing proxy'' and ``Outcome-proxy'' are generated linearly as in the strong-proxy setting. In the ``Treatment-inducing proxy'' setting, treatment assignment additionally depends on $\bfZ_i$ through $\delta_Z \bar Z_i$, with $\delta_Z = 0.8$ and $\delta_W = 0$. Thus, $\bfZ_i$ directly affects treatment assignment but has no direct effect on the outcome. This corresponds to the proximal DAG with the $\bfZ \to A$ edge present and the $\bfW \to Y$ edge omitted. In the ``Outcome-proxy'' setting, we set $\delta_Z = \delta_W = 0.8$, so that $\bfZ_i$ directly affects treatment assignment and $\bfW_i$ directly affects the outcome. 

For the design-assisted estimators, we consider two constructions of the nuisance feature map. The first is a raw-proxy construction, used for the same three design-assisted estimators as in Section 5.1. Let $\widehat{\bfR}_i$ denote the leading sparse principal-component scores computed from the centered and scaled baseline covariates $\bfX_{i,0}$. We set
\ba
H_{\mathrm{raw}}(\bfX_i;\hat S) = \left( \bfZ_i^\top, \, \bfW_i^\top, \, \widehat{\bfR}_i^\top \right)^\top.
\ea
This construction treats the proxy variables and low-dimensional design summaries as a flexible dictionary for approximating the latent confounding component.

The second construction is motivated by proximal causal learning \citep{tchetgen2024introduction}. For each coordinate
$k = 1, \ldots, q_W$, we estimate
\ba
g_k(A_i, \bfX_{i,0}, \bfZ_i) = \E^\circ(W_{ik} \mid A_i, \bfX_{i,0}, \bfZ_i)
\ea
using a first-stage high-dimensional LASSO regression of $W_{ik}$ on $(A_i, \bfX_{i,0}^\top, \bfZ_i^\top)^\top$, leaving $A_i$ unpenalized. The fitted values are denoted by
\ba
\widehat W_{c,ik} = \widehat g_k(A_i, \bfX_{i,0}, \bfZ_i),
\qquad
\widehat{\bfW}_{c,i} = (\widehat W_{c,i1}, \ldots, \widehat W_{c,iq_W})^\top .
\ea
We then define the proximal-control design-assisted feature map as
\ba
H_{\mathrm{pc}}(\bfX_i; \hat S) = \left( \widehat{\bfW}_{c,i}^\top, \, \widehat{\bfR}_i^\top \right)^\top.
\ea
This mirrors the proximal-control variable $\bfW_c = \E(\bfW \mid A, \bfX, \bfZ)$ used in proximal two-stage least squares, but here $\widehat{\bfW}_{c,i}$ is incorporated as a regularized nuisance feature rather than through an unpenalized instrumental-variable estimator. 

\clearpage

\section{Results for real data analysis}
\begin{table}[h!] 
\centering
\caption{Biological categories of the 37 genes selected 
by $\hat \bfbeta$.}
\label{table_select_genes}
\scriptsize
\begin{tabular}{llp{5.5cm}}
\toprule
\textbf{Category} & \textbf{Selected Genes} 
& \textbf{Biological Relevance} \\
\midrule
G-MDSC markers (3) 
& \textit{S100A12}, \textit{CD177}, \textit{MCEMP1} 
& Core NMF5 markers; top features in original 
  LASSO model \\
Longitudinal signal (1) 
& \textit{SERPINB2}
& Diverging trajectory between severe and 
  non-severe patients \\
Neutrophil chemotaxis (1) 
& \textit{CXCL2} 
& Neutrophil migration pathway enriched 
  in severe disease \\
MHC class II (2) 
& \textit{HLA-DMA}, \textit{HLA-DMB} 
& Enriched in non-severe patients \\
Interferon response (4) 
& \textit{RARRES3}, \textit{IRF2BP2}, 
  \textit{UPP1}, \textit{CTSW} 
& Interferon pathways enriched in infection 
  and associated with severity \\
Metabolism / ROS (3) 
& \textit{NQO2}, \textit{GPAT4}, \textit{ALDH1A1} 
& Metabolic pathways associated with 
  severe disease \\
Myeloid differentiation (2) 
& \textit{HOTAIRM1}, \textit{PDE4D} 
& Known roles in neutrophil development 
  and activation \\
Humoral immunity (2) 
& \textit{IGHM}, \textit{IGLVI-70}$^*$
& Humoral immune responses associated 
  with severity \\
Other protein-coding (11) 
& \textit{CCR3}, \textit{TSPAN2}$^*$, 
  \textit{TDRD9}, 
& Immune signaling, regulation, \\
& \textit{USP11}, \textit{METTL7B}, \textit{COA4}, 
& and cell metabolism \\
& \textit{YBEY}, \textit{MLH3}, 
  \textit{RPGRIP1}, & \\
& \textit{NOV}, \textit{AC008763.3}$^*$ & \\
Non-coding / 
& \textit{LINC00877}, \textit{RPL21P44},  
& Functional interpretation limited \\
unannotated (8)
& \textit{AC105749.1}, \textit{AC037198.2}, & \\
& \textit{AC108134.3}, \textit{AC023906.5}, & \\
& \textit{AC009303.4}, 
  \textit{ENSG00000288064} & \\
\bottomrule
\multicolumn{3}{l}{\footnotesize Ensembl IDs are retained for genes without established gene symbols.} \\
\end{tabular}
\end{table}

\bibliographystyle{apalike}
\bibliography{ref}

\end{document}